\documentclass[aps,10pt,onecolumn,floatfix,superscriptaddress,notitlepage]{revtex4-1}
\usepackage{amsmath,amsthm,amssymb,graphicx,bbm}
\usepackage[english]{babel}
\usepackage[]{algcompatible}
\usepackage{newfloat}

\newtheorem{proposition}{Proposition}
\newenvironment{definition}[1][Definition]{\begin{trivlist}
\item[\hskip \labelsep {\bfseries #1}]}{\end{trivlist}}

\renewcommand{\c}{\alpha} 
\renewcommand{\b}{\beta} 

\DeclareFloatingEnvironment[
    fileext=loa,
    listname=List of Algorithms,
    name=ALGORITHM,
    placement=tbhp,
]{algorithm}

\algloopdefx{RETURN}[1][]{\textbf{return} #1}

\newif\iffancy
\fancyfalse

\begin{document}

\title{
Classifying 50 years of Bell inequalities
}
\author{Denis Rosset}
\email{denis.rosset@unige.ch}
\affiliation{Group of Applied Physics, University of Geneva}
\author{Jean-Daniel Bancal}
\email{jdbancal.physics@gmail.com}
\affiliation{Center for Quantum Technologies, University of Singapore}
\author{Nicolas Gisin}
\affiliation{Group of Applied Physics, University of Geneva}

\date{\today}

\begin{abstract}
Since John S. Bell demonstrated the interest of studying linear combinations of probabilities in relation with the EPR paradox in 1964, Bell inequalities have lead to numerous developments. Unfortunately, the description of Bell inequalities is subject to several degeneracies, which make any exchange of information about them unnecessarily hard. Here, we analyze these degeneracies and propose a decomposition for Bell-like inequalities based on a set of reference expressions which is not affected by them. These reference expressions set a common ground for comparing Bell inequalities. We provide algorithms based on finite group theory to compute this decomposition. Implementing these algorithms allows us to set up a compendium of reference Bell-like inequalities, available online at http://www.faacets.com~. This website constitutes a platform where registered Bell-like inequalities can be explored, new inequalities can be compared to previously-known ones and relevant information on Bell inequalities can be added in a collaborative manner.
\end{abstract}

\maketitle


\section*{Introduction} 
\iffancy\addcontentsline{toc}{section}{\hspace*{-\tocsep}Introduction}\fi 

The years 1990’s started with the seminal paper presenting the Ekert'91 protocol~\cite{Ekert91}, relating quantum nonlocality to secure communication. This changed the world as far as quantum nonlocality is concerned; the study of Bell inequalities became respectable. So far not much was known beyond the famous CHSH inequality~\cite{chsh}. Here it is noteworthy to mention that Bell's original inequality published in 1964~\cite{Bell64} is not a Bell inequality in the modern sense, because it relies on the additional assumption of perfect anti-correlation when both sides perform the same measurement. In particular, little was known when the parties perform measurements with more than two possible outcomes. Kaszlikowski and co-workers performed numerical searches for experimental scenarios more resistant to noise~\cite{Kaszlikowski00}; this effort led Dan Collins, then at Geneva University, and colleagues to find the family of inequalities behind Kaszlikowski et al. finding, today known as the CGLMP inequalities~\cite{CGLMP04}. Meanwhile, Pitowski and Svozil, building on their understanding that the set of local correlations constitues a polytope, could find all the inequalities corresponding to the facets of two scenarios of interest~\cite{Pitowsky}. In a subsequent work, Sliwa~\cite{Sliwa} and Collins-Gisin~\cite{Collins} grouped the results of Pitowski and Svozil into families of inequalities equivalent under relabelings. In particular Sliwa found all the families corresponding to the scenario with 3 parties and binary inputs and outcomes, while Collins-Gisin found, among others, the family known as $I_{nnmm}$. Avis, Imai, Ito and Sakasi found many more Bell inequalities using specialized cut-polytopes~\cite{Avis}. And so the field expanded very significantly, though it would still be nice to have more families of inequalities valid for arbitrary number of parties, measurement settings and outcomes~\cite{Mermin90,Ardehali92,Belinskii93,Avis06,Bancal12}. Also, experiments on Bell's inequalities went out of the lab and entered applied physics~\cite{Tittel98,Jennewein00,Naik00,Tittel00}.
 
Another trend that started was the use of Bell-like inequalities to study the resources required to reproduce quantum correlations. Such resources should involve all parties at hand, as highlighted by the inequality proposed by Svetlichny back in 1987~\cite{Svetlichny}, first violated in 2009~\cite{Lavoie09}. Considering a bipartite situation, Bacon and Toner derived in 2003 some inequalities satisfied by all correlations that can be reproduced with shared randomness (as standard Bell inequalities) augmented by one single bit of communication~\cite{Bacon03Aux}. The fact that their two inequalities could not be violated by two entangled qubits motivated them to find a model of maximally entangled pairs of qubits using a single bit of communication, the nowadays famous Toner-Bacon model~\cite{Toner03Model}. A bit later Brunner and Gisin found inequalities valid for all correlations that can be simulated with one PR box~\cite{Brunner06}; this shows that some correlations corresponding to very partially entangled pairs of qubits can definitively not be simulated with a single PR-box, though the problem remains open both for medium entangled qubits and for the case of a single bit of communication and arbitrary entanglement.
 
Recently, Bell-like inequalities were also used in several contexts worth mentioning. The first context is the one of Entanglement Witnesses (EW). It is well-known that any violation of a Bell inequality witnesses entanglement (at least according to today's physics). Conversely, in the bipartite case, all EW written in a form independent of explicit observables -- that is written in a device-independent manner -- are also Bell inequalities. Hence, for two parties Bell inequalities are equivalent to Device-Independent Entanglement Witnesses (DIEWs). But for more parties this is no longer true: all Bell inequalities are not DIEWs~\cite{diew}, see also the recent experimental demonstration~\cite{Barreiro}. Second, in the context of randomness analysis, Bell inequalities can certify intrinsic randomness~\cite{ColbeckThesis,Pironio10,Colbeck12,Gallego13}. Third, the tool of Bell-like inequalities can be used to study hypothetical models of quantum correlations based on ``hidden influences'' propagating at finite-but-supraluminal speeds~\cite{Bancal12Hidden}. Finally, in the context of self-testing, violation of a Bell inequality can provide certification for the proper behavior of a device without relying on previous calibration~\cite{MayerYao04, Reichardt13, Yang14}. These recent developments show the relevance of finding a common language for our community to discuss its findings, as presented in this paper.

While it is quite straightforward to write down a Bell inequality, a number of parameters make this writing not unique. Thus, two inequalities with similar properties can look superficially very different. This degeneracy can hide obvious facts, and thus constitutes a practical obstacle in the study of Bell inequalities. As an example, the inequality A1, given in 2012 by Grandjean et al~\cite{Grandjean12}, is equivalent to an inequality published 8 years before as Eq. 4 in~\cite{Acin04}, yet this fact was not noticed at the time of publication.

We present here a scheme to deal with these redundancies, which allows each family of inequalities to be referenced by its index in a list of canonical inequalities. An open-source library implements our scheme. It can be freely used by researchers to automate the computations. Thanks to this tool, we launched a growing interactive library of Bell inequalities, available at the URL \textit{faacets.com}.

Our paper is structured as follows: we first clarify in Section~\ref{sec:bell-stuff} the concept of Bell inequalities and Bell expressions, before describing in Section~\ref{sec:redundancies} several degeneracies that can appear in the description of a Bell inequality. In Section~\ref{sec:cure}, we show how to remove with each of these degeneracies individually. This leads us to propose in Section~\ref{sec:decomposing} a method to decompose Bell inequalities into a canonical form. In Section~\ref{sec:compendium}, we describe the tools we are making available to decompose and classify Bell inequalities.

%
%
%
%

\section{Bell scenarios, Bell-like inequalities and oriented Bell expressions} 
\label{sec:bell-stuff}
In a Bell experiment, $n$ parties each hold a system that they measure successively with one of several measurement settings, each time recording one out of several possible measurement outcomes. In general, the number of available measurements might differ from one party to another one, just like the number of outcomes that these measurements can produce. These numbers of measurement settings and outcomes, together with the number of parties taking part in the experiment, define a \textit{Bell scenario}. For simplicity, we consider in the main text that all parties have the same number of possible settings and outcomes, respectively $m$ and $k$. We call these scenarios \textit{homogeneous}, and refer to them with the triple $(n,m,k)$. Except for one additional step that needs to be taken into consideration (c.f. Appendix~\ref{app:nonhomogeneous}), all the results contained in the main text extend straightforwardly to non-homogeneous scenarios.

\begin{figure}
\includegraphics[width=0.2\textwidth]{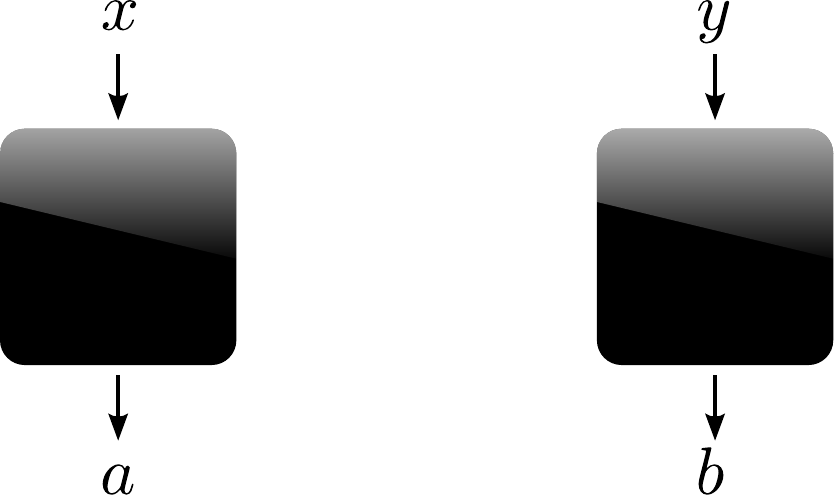}
\caption{Schematic representation of a typical bipartite Bell-type experiment: two parties have access to individual systems, which they can  probe with different measurement settings (as indexed by $x$ and $y$). Upon measurement, the systems provide outcomes $a$ and $b$. The estimation of the conditional probability $P(ab|xy)$  characterizing the response of the devices requires no knowledge of their inner workings~\cite{rmp}.}
\label{fig:setup}
\end{figure}

In a situation in which only information about the measurement settings and outcomes used by the parties is available, the conditional probability $P(ab|xy)$ with which outcomes $a$ and $b$ of the different parties (here two parties) are observed when they use measurement settings $x$ and $y$ respectively (see Figure~\ref{fig:setup}), is of particular interest. These probabilities (or \textit{correlations}) can indeed be estimated in principle simply by repeating the experiment a sufficient number of times, and without further assumption on the measured systems or measurement procedure (i.e. in the so-called device-independent manner~\cite{rmp,introThesis}). These correlations form a list of $D=(mk)^n$ real numbers that can be conveniently represented as a vector or point in the vector space $\Omega=\mathbb{R}^D$.

Properties of these points can be highlighted with the aid of \textit{Bell expressions}, i.e. linear forms 
\begin{equation}
\label{eq:cabxy}
B(P) = \sum_{abxy} \c_{abxy} P(ab|xy).
\end{equation}
Any such expression can be defined by its coefficients $\c_{abxy}$, which also form a vector in the dual vector space $\Omega^*=\mathbb{R}^D$.

A Bell expression taking a definite value $B(P)=v$ defines a hyperplane in the space of correlations which divides the space into two distinct regions. Such hyperplane can thus always be used to demonstrate that a point $P^*$ does not belong to some convex set $S\subset\Omega$, whenever it is the case~\cite{textbook}. This is conventionally done by writing a \textit{Bell-like inequality} $B(P)\leq \beta\ \forall P\in S$ with bound $\b$, and showing that the inequality is violated for the considered point of probabilities, i.e. $B(P^*)>\b$.

A set of particular interest in the space of conditional probabilities is the local set, given by all correlations which can be decomposed as
\begin{equation}\label{eq:locality}
P(ab|xy) = \int \rho(\lambda) P_A(a|x,\lambda) P_B(b|y,\lambda) d\lambda,
\end{equation}
where $\rho(\lambda)$ is positive and normalized~\cite{rmp}. Inequalities satisfied by this set are referred to as \textit{Bell inequalities}. Since this set is a polytope~\cite{rmp}, it can be described with a minimal number of such inequalities: the facets of this polytope. These Bell inequalities are thus of special interest.

Note that with respect to a given convex set, every Bell expression $B$ can give rise to two Bell-like inequalities: one bounding the expression from below, and one from above. In general, these the two inequalities can have different natures. For instance, one might be a facet of the local polytope, while the other one is not (c.f. Figure~\ref{fig:froissart}). We thus wish to distinguish between these two inequalities.

\begin{figure}
\includegraphics{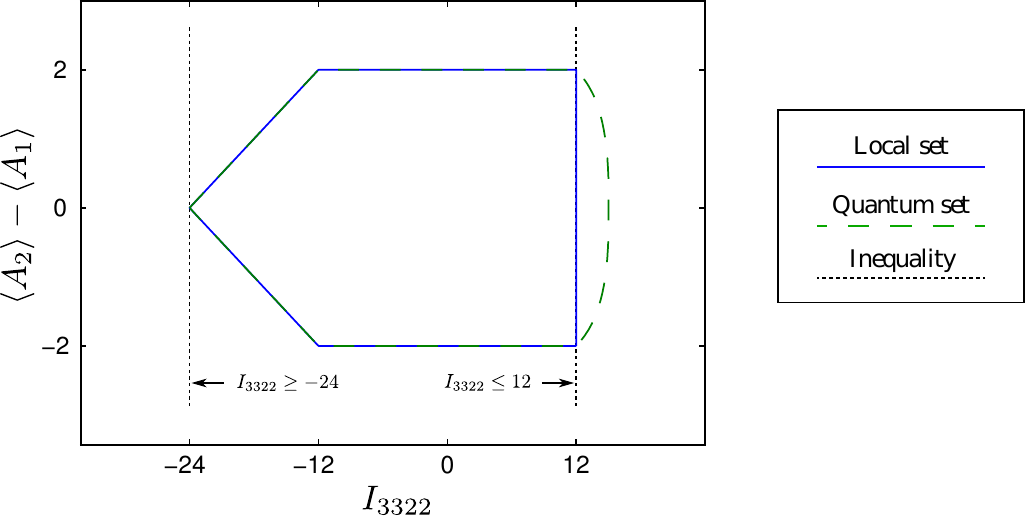}
\caption{Projection of the local (full line) and quantum (dashed line) sets of correlations in a particular two-dimensional probability subspace for the $(2,3,2)$ scenario. The horizontal axis is the value of the $I_{3322}$ Bell expression~\cite{Froissart81,faacetsI3322}. This Bell expression has different properties with respect to the sets displayed here depending on the orientation considered. On the right side of the figure, the $I_{3322}$ expression gives rise, through an upper bound, to a facet of the local polytope (dotted line), which admits a quantum violation. In the opposite direction, however, the same expression (dotted line) gives not rise to a facet, and shows no quantum advantage.}
\label{fig:froissart}
\end{figure}

At the same time, several other sets of correlations than the local one, and their associated Bell-like inequalities, have proven useful in various contexts (e.g. to demonstrate genuine multipartite nonlocality or entanglement, non-simulability, etc~\cite{Svetlichny,diew,Bacon03Aux}). A given Bell expression can then have various upper and lower bounds of interest. When comparing two inequalities given by coefficients $\c_{abxy}$, $\c'_{abxy}$ and, say, upper-bounds $\b$, $\b'$, it is thus important to recognize not only if they describe the same Bell-like inequality, but also if they represent the same expression with different bounds.

We thus wish to identify identical expressions with different bounds, but distinguish between expressions bounded from above or from below. For this purpose, we define an \textit{oriented Bell expression} to be the combination of a Bell expression $B$ with an inequality sign. By convention, we choose to consider only oriented Bell expressions bounded from above, i.e. with inequality sign ``$\leq$''. All upper bounds on $B$ can then be seen as properties of this object, while lower bounds, which are upper bounds on the Bell expression $-B$, are dissociated from it. Geometrically speaking, an oriented Bell expression can be understood as describing a direction of interest in the space of correlations. Opposite directions can have different properties (c.f. Figure~\ref{fig:froissart}).

Note that, since it describes a direction in probability space, an oriented Bell expression needs not come with a bound per se. At the same time, an oriented Bell expression with bounds can be understood simply as a collection of Bell-like inequalities. With respect to one chosen set of correlations, a Bell expression has a unique tight upper bound, so it only gives rise to one tight Bell-like inequality.

\section{Degeneracies in the description of Bell-like inequalities}
\label{sec:redundancies}
The sets of correlations one can wish to consider in a Bell experiment typically satisfy several constraints. Each constraint introduces freedom in the way that a Bell inequality can be written while yet performing the same test. Here we describe such constraints and the degeneracies they induce on the description of Bell-like inequalities. In the next section we propose a solution to lift these degeneracies.

All the constraints considered here are satisfied by the local, quantum and no-signaling sets of correlations.

\subsection{The normalized no-signaling subspace}
\label{sec:redundancies:no-signaling}
Maybe one of the most trivial constraints expected to be satisfied by all physical probabilities is that they are normalized. This can be expressed as
\begin{equation}\label{eq:normalization}
\sum_{ab}P(ab|xy)=1\ \forall x,y
\end{equation}
in the bipartite case (which we use by default in the rest of this paper for the sake of the example). Whenever the probabilities one wants to consider satisfy this constraint, we can rewrite any Bell expression in infinitely many different ways. As an example, consider the positivity constraint $-P(11|11)\leq 0$, with coefficients $\c_{abxy}=-\delta_{a1}\delta_{b1}\delta_{x1}\delta_{y1}$. The coefficients $\c'_{abxy}=\c_{abxy}+K(\delta_{x1}-\delta_{x2})$ define the same inequality through $\sum_{abxy}\c'_{abxy}P(ab|xy) \leq 0$, for any $K\in\mathbb{R}$. Therefore we see a first degree of freedom in the way one can write the simple positivity constraint.

Apart from considering normalized probabilities, we may also wish to restrict our attention to correlations satisfying the no-signaling condition. All correlations predicted by quantum theory indeed satisfy this constraint, and together with the normalization conditions, these constraints define the smallest affine probability subspace $\omega\subset\Omega$, of dimension $d=(1+m(k-1))^n-1$, which contains the set of quantum correlations. These conditions take the form
\begin{equation}\label{eq:nosignaling}
\sum_b P(ab|xy) = P(a|x)\ \forall y,
\end{equation}
and similarly for the sum over the first party's outcomes. Just like with normalization, the no-signaling condition again defines infinitely many variations in the way one can write a given Bell expression.

Let us point out an example of two famous inequalities which are equivalent over all normalized no-signaling probabilities: the CHSH inequality~\cite{chsh} was first described as
\begin{equation}\label{eq:chshinequality}
|E_{11}-E_{12}| \leq 2 - E_{21} - E_{22}
\end{equation}
where $E_{xy}=\sum_{a,b=1}^2 (-1)^{a+b}P(ab|xy)$ is a correlation function. Choosing sign $+$ for the absolute value, this inequality can be described by the coefficients $\c^\text{CHSH}_{abxy}=(-1)^{a+b+x(y+1)}$ and bound $\b^\text{CHSH}=2$. On the other hand, the CH inequality~\cite{ch} reads
\begin{equation}
\label{eq:chinequality}
P(11|11)-P(11|12)+P(11|21)+P(11|22)-P_A(1|2)-P_B(1|1) \leq 0
\end{equation}
where $P_A$ and $P_B$ are the parties' marginal probabilities. Writing $P_A(a|x)=\sum_bP(ab|x1)$ and $P_B(b|y)=\sum_aP(ab|1y)$, this expression can be described by the coefficients
\begin{equation}
\label{eq:chcoeffs} \c^\text{CH}_{abxy}=-\delta_{a2}\delta_{b1}\delta_{x1}\delta_{y1}-\delta_{a1}\delta_{b2}\delta_{x2}\delta_{y1}-\delta_{a1}\delta_{b1}\delta_{x1}\delta_{y2}+\delta_{a1}\delta_{b1}\delta_{x2}\delta_{y2},
\end{equation}
and bound $\b^\text{CH}=0$. These coefficients and bounds are clearly different from the CHSH ones and so do not appear to be related to those of CH at first sight. However, adding the following expression to it:
\begin{equation}
\big [ ( \delta_{x1} - \delta_{x2}  )(1 + 2(\delta_{b1} - \delta_{b2}) - (\delta_{y1} - \delta_{y2}) ) +  ( \delta_{y1} - \delta_{y2}  )(1 + 2(-1)^x(\delta_{a1} - \delta_{a2})) \big ]/8 = \c^\text{CHSH}_{abxy}/4 - \c^\text{CH}_{abxy} - \frac{1}{8},
\end{equation}
which vanishes for probabilities satisfying~\eqref{eq:nosignaling}, reveals the affine transformation that relates both expressions. The two inequalities thus define a test along the same hyperplane for all normalized no-signalling correlations.

Here we choose to consider as equivalent inequalities which act identically on the space of normalized no-signaling correlations. We thus wish to eliminate this kind of redundancies. Note that this might not be desired in some specific situations in which, for instance, communication between the parties is allowed.

\subsection{The relabelings}
\label{sec:redundancies:relabelings}

In a Bell experiment, it is often the case that no preferential importance is attached to any particular party, measurement setting or outcome. Indeed, the value of a particular party, measurement setting or outcome is often used simply as a label, attributed for the sake of distinguishability, but with a level of arbitrariness. Therefore, any permutation of parties, settings or outcomes which is compatible with the Bell scenario transforms probabilities $P(ab|xy)$ into $P'(ab|xy)$ which can be obtained from the same data, by relabeling the parties, settings or outcomes. In turn, the same permutation can be applied to any Bell expression whenever the considered set of correlations is also invariant under such permutations.

As an example, consider the inequality $P(12|11)\leq 1/2$. It describes a different half-space than $P(11|11)\leq 1/2$ in $\omega$. Yet, any experiment whose correlations violate one of these inequalities can also violate the other one if the labels of Bob's outcomes $1$ and $2$ are attributed in an opposite manner (say $1$ to horizontal photon polarization and $2$ to vertical polarization instead of the opposite, e.g.).

Since most sets one is concerned with are invariant under relabeling of parties, settings and outcomes, we wish to analyse these inequalities independently of such relabelings. All sets of correlations need not satisfy this constraint, though (see~\cite{Woodhead} for an example).

\subsection{Superfluous parties, settings or outcome distinction}
\label{sec:redundancies:superfluous}

Having considered conditions that apply to sets of correlations in a fixed Bell scenario, we now consider conditions that one can expect to hold in the relation between different scenarios. In this context, we refer to a rule that generates sets of correlations for various scenarios as a \textit{model}. For instance, the local model, defined by Eq.~\eqref{eq:locality}, generates distinct sets of correlations in each Bell scenario.

All the constraints considered here are again satisfied by the local, quantum and no-signaling sets of correlations.

\subsubsection{Superfluous parties}
First, consider an inequality which doesn't involve certain parties, even though they are available in the considered scenario. This would be the case if one were to test the CHSH inequality in a tripartite experiment for instance. Clearly, the coefficients of this inequality are not identical to the ones of the CHSH inequality for two parties (they even belong to different spaces). Yet, the test performed is arguably physically identical. One is thus tempted to neglect the third irrelevant party from the scope, and rewrite the inequality in a bipartite scenario only. This simple operation brings us back to analyze an inequality in a simpler scenario.

While this operation sounds trivial, it is justified to carry the bound of the inequality through it only when the respective sets of correlations defined by the tested model in both scenarios satisfy some constraints. Namely, they must be such that the set $\mathcal{P}_{n',n}$ of $n$-partite correlations produced by the model in an $n'$-partite situation, where $n'>n$, coincides with the set $\mathcal{P}_{n}$ of $n$-partite correlations it produces in presence of $n$ parties, i.e.
\begin{equation}
\mathcal{P}_{n',n} = \mathcal{P}_n.
\end{equation}
This is of course the case for most sets of interest. We thus wish to neglect parties of a Bell scenario which do not intervene in a Bell test.

\subsubsection{Superfluous measurement settings}
When the value of a Bell expression does not depend on which result a party chooses to output for some setting $x^*$, i.e. $c_{abx^*y} = c_{a'bx^*y}\ \forall a,a',b,y$, then it should be clear that an experiment evaluating it could in principle be achievable without using this measurement setting $x^*$ at all. Thus, we also consider removing such setting to simplify the scenario. Again, this is valid whenever the tested model satisfies the condition
\begin{equation}
\mathcal{S}_{m',m} = \mathcal{S}_m,
\end{equation}
where $\mathcal{S}_m$ is the set of correlations produced by the considered model with $m$ possible settings, and $\mathcal{S}_{m',m}$ the set achieved when $m'>m$ settings are used, but the statistics of the $m'-m$ additional settings are neglected.


\subsubsection{Superfluous outcome distinction}
When no setting is superfluous, but yet two outcomes play the same role in a Bell-like inequality, i.e. $\exists$ $a_1 \ne a_2$, $x^*$, such that $c_{a_1bx^*y}=c_{a_2bx^*y}$ for all $b,y$, one might as well not distinguish between them and just assign a single outcome for both cases. Indeed, arbitrarily many different inequalities can be generated by increasing the number of outcomes which together share the same probability weight, while the test performed by the corresponding inequality remains the same because it does not distinguish between them. We thus wish to avoid such degeneracy as well. This is possible whenever the considered set satisfies
\begin{equation}
\mathcal{O}_{k',k} = \mathcal{O}_k,
\end{equation}
where $\mathcal{O}_k$ is the set of correlations produced by the considered model with $k$ possible outcomes, and $\mathcal{O}_{k',k}$ the set achieved when the model can produce $k'>k$ outcomes, but some are grouped together to form only $k$ of them.

\subsubsection{Note on liftings}
As we just argued, the bound of an inequality is unaltered in presence of superfluous parties, inputs or output distinction, provided the tested model satisfies the corresponding constraint. In some cases, however, more can be said about the relationship between inequalities created by adding artificial parties, settings or outcomes distinction.

One such observation was presented in~\cite{liftings}, where it was shown that the property of an inequality being (or not) a facet of the local polytope is preserved when adding irrelevant settings or distinction between outcomes, an operation also known as \textit{lifting} an inequality. This property is however not kept when extending the experiment to another party in the way we just described (c.f. Figure~\ref{fig:liftings}). Rather, lifting an inequality to more parties in such a way that its facet property is preserved (with respect to the local polytope) can be accomplished by conditioning the test that this inequality performs to some outcome observed by the additional parties~\cite{liftings}. This operation is a special case of the one we introduce now.

\begin{figure}
\includegraphics[width=0.5\textwidth]{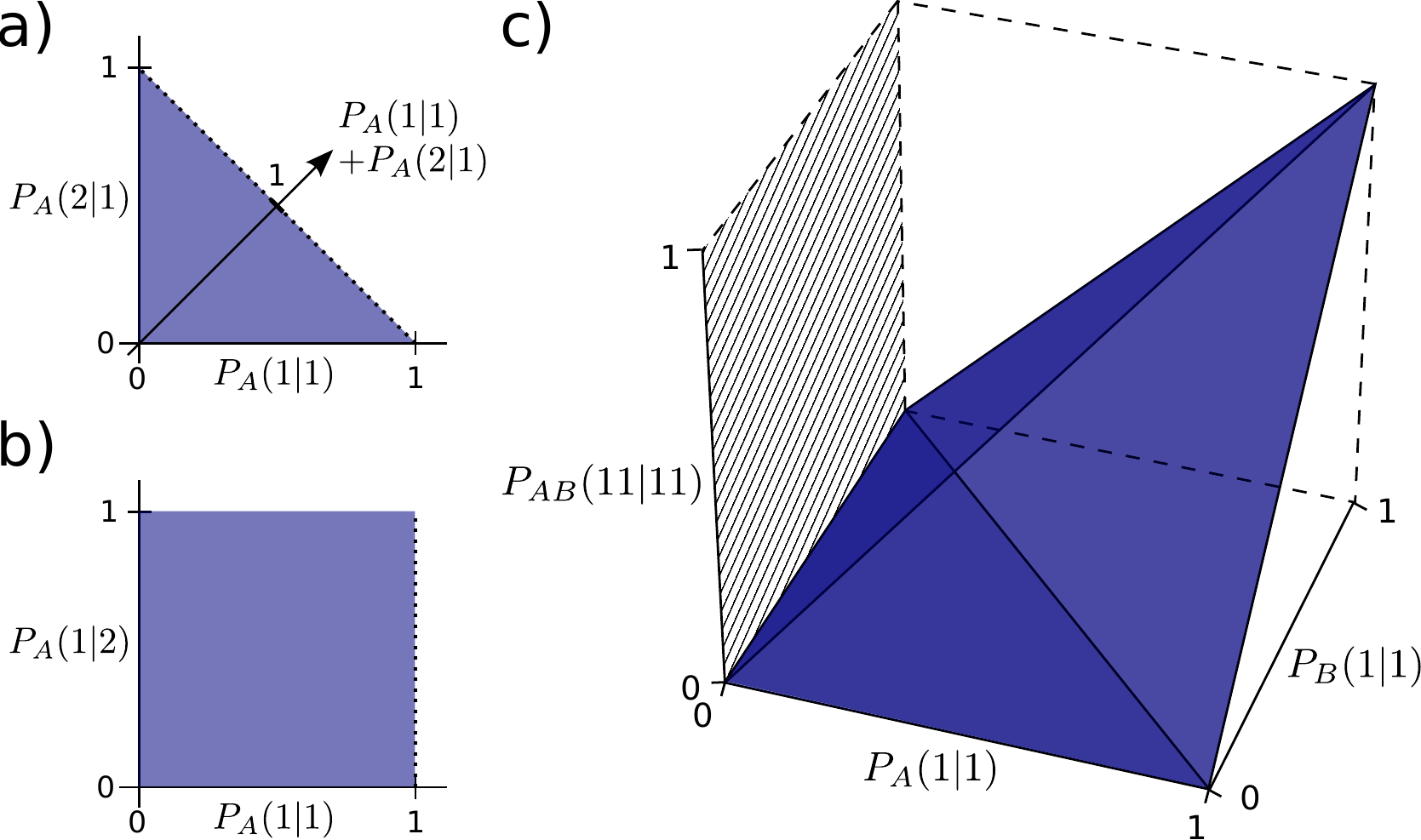}
\caption{Illustration of i/o-liftings and inequality composition. This figure shows the effect of adding (a) an output distinction, (b) an input or (c) a party in the space of correlations, for the simplest (trivial) scenarios. Allowed correlations are within the shaded areas/volumes. (a) When distinguishing between outcomes 1 and 2, the boundary of the set of possible correlations is given by the boundary of the 1-outcome set corresponding to grouping the two outcomes into one, plus the usual positivity conditions. (b) When adding one input, the same inequalities as previously still define the boundary of the new set. (c) Adding a party necessarily adds more than one dimension to the probability space. The boundary of the new set involves nontrivial combinations of all variables. In this case the new inequalities can be seen as \textit{compositions} of the initial boundaries of the respective 1-party sets. The boundary of the 1-party set expressed in the new space (dashed here) is not a facet of the new one (shaded).}
\label{fig:liftings}
\end{figure}

\subsection{Composite inequalities}
\label{sec:redundancies:composite}

Here, we describe a way of adding parties in a Bell scenario which preserves the facet property for the local polytope. The inequality obtained is tight for all models satisfying the constraint given below. This immediately implies that the local, quantum and no-signaling bounds are inherited from this construction.

As mentioned, adding a passive party to a Bell test doesn't result in an optimal test in the new extended scenario (see figure~\ref{fig:liftings}), but conditioning the test to an outcome of the additional party does, and is known as a lifting~\cite{liftings}. Let us thus consider models which satisfy the constraint that the correlations they produce in an $n$-partite scenario coincide with the $n$-partite correlations produced in an $n'$-partite scenario, with $n'>n$, whenever these correlations are conditioned to the $n'-n$ remaining parties outcomes. We denote this condition as:
\begin{equation}\label{eq:conditionality}
\mathcal{P}_{n',n|n'-n}=\mathcal{P}_n.
\end{equation}
In other words, this condition states that any $n$-partite correlations can be prepared in an $n'$-partite scenario upon heralding from the results of measurements performed by the $n'-n$ remaining parties; and that any such preparation produces valid $n$-partite correlations.

We show in Appendix~\ref{compositeAppendix} that if two expressions with coefficients $\c_{abxy}$ and $\c'_{cdzw}$ are bounded above by $\b_+$ and $\b'_+$ and below by $\b_-$ and $\b'_-$ for such models (satisfying \eqref{eq:conditionality}), then the tensor product of the two inequalities satisfies:
\begin{equation}\label{eq:prodExpression}
\sum_{abcdxyzw}\c_{abxy}\c'_{cdzw}P(abcd|xyzw) \leq \max(\b_- \b'_-,\ \b_- \b'_+,\ \b_+ \b'_-,\ \b_+ \b_+).
\end{equation}
As an illustration, consider bounding the tensor product of two CHSH expressions with respect to the local set. The corresponding expression reads
\begin{equation}\label{eq:CHSHCHSH}
\sum_{abcdxyzw}\c_{abxy}^{CHSH}\c_{cdzw}^{CHSH}P(abcd|xyzw)=\sum_{abxy}\c_{abxy}^{CHSH}P(ab|xz)\sum_{cdzw}\c_{cdzw}^{CHSH}P(cd|zw,abxy)
\end{equation}
where we use the fact that correlations are no-signaling. Since the set of local correlations satisfies~\eqref{eq:conditionality}, one has that $\sum_{cdzw}\c_{cdzw}^{CHSH}P(cd|zw,abxy)=\sum_{cdzw}\c_{cdzw}^{CHSH}P(cd|zw)$, which is bounded between $-2$ and $2, \ \forall\ a,b,z,w$. The local bound of~\eqref{eq:CHSHCHSH} is thus $2\cdot 2=4$ in agreement with~\eqref{eq:prodExpression}.

In Appendix~\ref{compositeAppendix} we also show that any two Bell-like inequalities defining facets with respect to a model can be compose to produce a new facet for this model. Thus arbitrarily many Bell inequalities can be generated again in this way by composing Bell inequalities involving fewer parties. Since this construction carries through significant properties for the local, quantum and no-signaling sets, and provides insight into inequalities that can be seen as products of simpler expressions, we choose also not to consider composite inequalities as canonical. This includes not considering liftings of an inequality to more parties as fundamental, since these can be seen to be compositions with the positivity constraints (construction~\eqref{eq:prodExpression} generalizes to compositions involving different number of parties.). Note that expressions with superfluous parties are also composite, as compositions with a constant. They are thus also detected as non-canonical here.

Examples of properties which are not inherited from rule~\eqref{eq:prodExpression} include for instance Svetlichny and biseparable bounds~\cite{Svetlichny,diew}, because their corresponding sets don't satisfy~\eqref{eq:conditionality}.

\section{Removing the degeneracies}
\label{sec:cure}
We now describe a cure for each degeneracies mentioned in the last section. Taken together, this allows us to identify Bell-like inequalities and expressions independently of any such degeneracy. In the next Section~\ref{sec:decomposing} we will use this to define families of Bell-like inequalities and write down a decomposition for any oriented Bell expression in terms of canonical representatives.

\subsection{The normalized no-signaling subspace}
\label{sec:cure:no-signaling}

An original way to deal with the degeneracies induced by the normalization and no-signaling conditions was provided in~\cite{Collins}: it consists in parametrizing the space of probabilities with joint and marginal probabilities, but without monitoring the last outcome. Every coefficient involving the last outcome of some parties can then be computed from the normalization and no-signaling condition, and any Bell expression is described by a unique set of coefficients. We refer to this as the Collins-Gisin parametrizations. While it solves part of the problem, this solution does not treat all outcomes similarly. As a result, computing the effect of relabelings that involve the last outcomes require the use of arithmetic, and this would complicate significantly the search for a particular representative under relabelings.

To avoid this complication, we choose instead to keep all the coefficients in \eqref{eq:cabxy}, such that relabeling parties, settings or outcomes only amounts to permuting the coefficients of $\c$. The normalization and no-signaling redundancies can then be eliminated in a way which is compatible with these permutations by choosing a parametrization of Bell expressions acting on normalized no-signaling subspace $\omega$ in $\Omega$ which treats all parties, settings and outcomes on an equal footing.

\subsubsection{A complete basis symmetric under relabelings}
\label{sec:cure:no-signaling:basis}

To construct such a parametrization, we start by identifying components of the dual space $\Omega^*$ that are symmetric under relabellings and capture the normalization and no-signaling conditions. We eventually wish to extend our construction to an arbitrary number of parties, so we choose to consider bases that are tensor across the parties. We thus only need to define a basis for single parties. In these terms, the normalization \eqref{eq:normalization} can be written as:
\begin{equation}
\label{eq:normalizationindec}
\sum_{axby} \mu_{ax} \mu_{by} P(ab|xy) = 1, \quad \mu_{ax} = \mu_{by} = \frac{1}{m}.
\end{equation}
Here the value of the components $\mu_{ax}$ is fixed by the constraint that it must be symmetric under permutation of the inputs. We thus have isolated the component of $\Omega^*$ which will encode the normalization. Adding $\kappa$ times $\mu_{ax}\mu_{by}$ to a Bell expression shifts its value by the constant $\kappa$ on all normalized correlations.

We now proceed with the no-signaling equations \eqref{eq:nosignaling}, which can be rewritten as:
\begin{equation}
\label{eq:nosignalingindec}
\begin{split}
&\sum_{by} u_{by} \sum_{ax} \nu^\xi_{ax} P(ab|xy) = 0,\ \forall\ u_{by}, \quad \nu^\xi_{ax} = \delta_{x, \xi} - \delta_{x, \xi+1}, \quad \xi=1\ldots m-1,
\end{split}
\end{equation}
and similarly for no-signaling from Bob to Alice. This time, the individual components $\nu^\xi_{ax}$ are not invariant under permutation of settings. However, one can verify that the subspace $\text{span}(\{\nu^\xi_{ax}\}_\xi)$ they generate is invariant: any permuted $\nu'_{ax} = \delta_{x,\xi'} - \delta_{x,\xi''}$ can be re-expressed as a linear combination of $\nu^{\xi}_{ax}$ only.

Together, $\{ \nu^\xi_{ax} \}$ and $\mu_{ax}$ define a basis for an $m$-dimensional subspace of $\Omega^*$ which is invariant under permutations of settings and outcomes for Alice. To form a complete basis, $m(k-1)$ additional elements $\{ \lambda^i_{ax} \}_{i=1...m(k-1)}$ are needed, and the subspace spanned by these elements should also be invariant under relabelings.  A simple form for these elements is:
\begin{equation}
\label{eq:lambda}
\lambda^{\zeta \xi}_{ax} = \delta_{x,\xi} \left ( \delta_{a,\zeta} - \delta_{a,\zeta+1} \right )
\end{equation}
for $\xi=1...m$ and $\zeta=1...k-1$. This choice generalizes the correlators already used in the literature in the case of scenarios with binary outcomes.


A bipartite Bell expression can then be expressed in terms of a complete basis $\{u\} = \{ \mu \} \cup \{ \lambda^{\zeta \xi} \}_{\zeta=1...k-1, \xi=1...m} \cup \{ \nu^\xi \}_{\xi=1...m-1}$ as 
\begin{equation}
c_{abxy} = \sum_{ij=1}^{mk} \gamma_{ij} u^i_{ax} u^j_{by},
\end{equation}
where $\gamma_{ij}$ are the components of the expression in the new symmetric basis.

\begin{figure}
\includegraphics{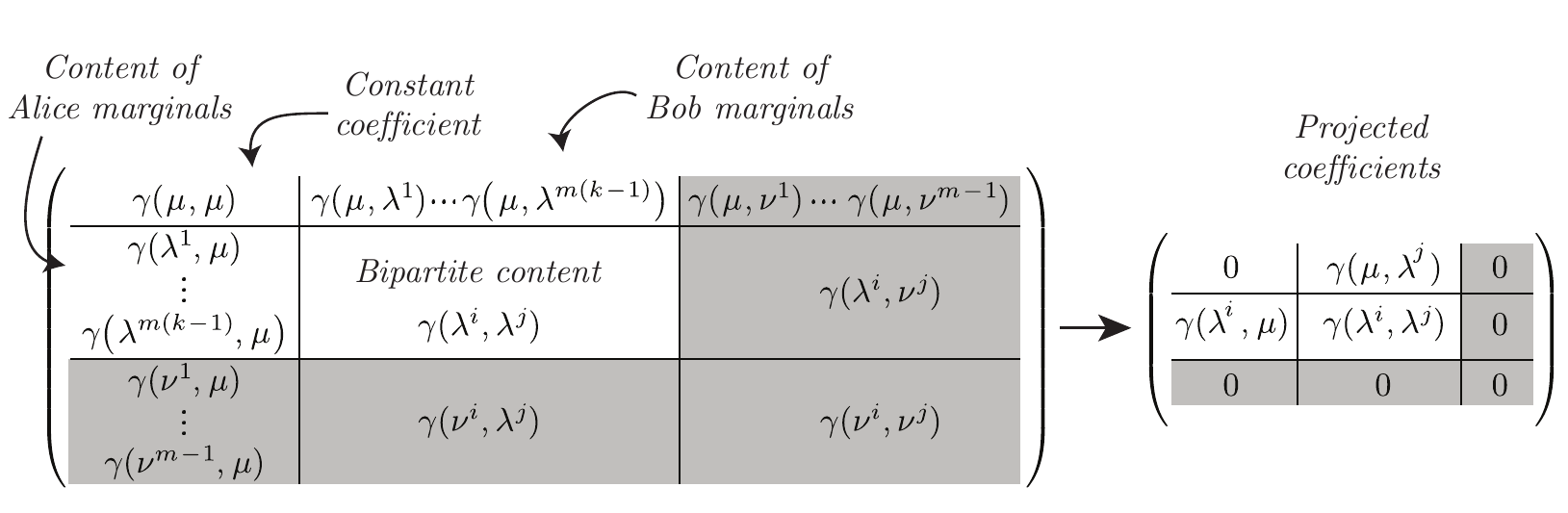}
\caption{Matrix of coefficients for the Bell expression $\c_{abxy}$ expressed in the basis that isolates no-signalling terms, and is compatible with the Bell permutations. We write $\gamma(u,v)$ the coefficient $\gamma$ corresponding to the basis elements $u$ and $v$. Terms describing marginals and no-signaling conditions are isolated. The latter are shaded in gray. The projection shifts the normalization coefficient to the bound, and removes coefficients corresponding to no-signaling constraints.}
\label{fig:coefficientsmatrix}
\end{figure}

Since adding any terms of the form $\nu^\xi_{ax}u^j_{by}$ or $u^i_{ax}\nu^\xi_{by}$ to the expression does not change its value for no-signaling correlations, we can always choose to add such terms so that their coefficients disappear in the expression (c.f. Figure~\ref{fig:coefficientsmatrix}). Similarly, adding $\mu_{ax}\mu_{by}$ shifts the whole expression by a constant. We thus define an expression independently of such shifts by setting this element to zero. Re-expressing the projected coefficients in terms of probabilities, this defines new coefficients $\overline{\c}_{abxy}$ which are not subject to normalization or no-signaling redundancies, and behave well under permutations or parties, settings and outcomes:
\begin{equation}\label{eq:cbar}
\overline{\c}_{abxy} = \left(\sum_{ij=1}^{m(k-1)+1} \gamma_{ij}  u^i_{ax} v^j_{by} \right) - \gamma_{11} \mu_{ax} \mu_{by}.
\end{equation}
The inequality $\sum_{abxy} \c_{abxy} P(ab|xy) \le \beta$ can then be rewritten as:
\begin{equation}
\label{eq:afterprojection}
\sum_{abxy} \overline{\c}_{abxy} P(ab|xy) \le \beta - \gamma_{11}.
\end{equation}

A last degeneracy is given by multiplying both sides of \eqref{eq:afterprojection} by some factor. For inequalities with rational coefficients, this can be dealt with by multiplying the coefficients $\overline{\c}_{abxy}$ by a positive number so that the resulting coefficients are integers with greatest common divisor 1. This allows one to identify any oriented Bell expression uniquely within the normalized no-signaling subspace. Moreover, by construction this identification is compatible with Bell permutations: a permutation of parties, settings or outcomes applied on the coefficients $\overline{\c}_{abxy}$ never re-introduces terms of the form $\nu^\xi_{ax}u^j_{by}$, $u^i_{ax}\nu^\xi_{by}$ or $\mu_{ax}\mu_{by}$.

The construction above generalizes readily to multipartite scenarios, by repeating the basis construction for the additional parties. In the Appendix~\ref{app:nonhomogeneous}, we give an example of this construction for a non-homogeneous scenario.

\subsubsection{Example: the CH and CHSH expressions in the normalized no-signaling subspace}
\label{sec:cure:no-signaling:example}

To illustrate our construction, and the uniqueness of the Bell expression it produces, let us take the coefficients $\c^\text{CH}$ for the CH expression~\eqref{eq:chcoeffs}, and decompose them on the basis $u^i_{a|x}$:
\begin{equation}
\{ \vec{u}^i \} = \{ \vec{\mu}, \vec{\lambda}^{11},  \vec{\lambda}^{12}, \vec{\nu}^{1} \} =
 \left \{ \left ( \begin{array}{r} u^i_{1|1} \\ u^i_{2|1} \\  u^i_{1|2} \\ u^i_{2|2} \end{array} \right ) \right \} = 
\left \{
  \left ( \begin{array}{r} 1/2 \\ 1/2 \\  1/2 \\ 1/2 \end{array} \right ),
  \left ( \begin{array}{r} 1 \\ -1 \\ 0 \\ 0 \end{array} \right ),
  \left ( \begin{array}{r} 0 \\ 0 \\  1 \\ -1 \end{array} \right ),
  \left ( \begin{array}{r} 1 \\ 1 \\ -1 \\ -1 \end{array} \right )
\right \},
\end{equation}
and using the same construction for Bob's basis.
Let us take the Bell expression for the CH inequality in~\eqref{eq:chinequality}, and let us decompose:
\begin{equation}
\gamma_\text{CH} = \frac{1}{8} \left( \begin{array}{r|rr|r}
-4  & 0 & 0 & -2 \\
\hline
0  & 2 &-2 & 2 \\
0  & 2 & 2 & -2 \\
\hline
-2 & -2 & -2 & 1
\end{array} \right ), \qquad
\overline{\gamma}_\text{CH} = \left( \begin{array}{r|rr|r}
0  & 0 & 0 & 0 \\
\hline
0  & 1 &-1 & 0 \\
0  & 1 & 1 & 0 \\
\hline
0 & 0 & 0 & 0
\end{array} \right ),
\end{equation}
such that:
\begin{equation}
\label{eq:chshprojected}
B(P) = \sum_{abxy} \overline{\c}_{abxy} P(ab|xy) \le 2, \quad
\overline{\c}_{abxy} = \sum_{ij} \overline{\gamma}_{ij} u^i_{ax} u^j_{by} = (-1)^{a+b+x(y+1)}.
\end{equation}

When considering the CHSH inequality~\eqref{eq:chshinequality}, we observe that the correlators $E_{\xi \upsilon} = \sum_{abxy} \lambda^\xi_{ax} \lambda^\upsilon_{by} P(ab|xy)$, and thus:
\begin{equation}
\gamma_\text{CHSH} = \left( \begin{array}{r|rr|r}
0  & 0 & 0 & 0 \\
\hline
0  & 1 &-1 & 0 \\
0  & 1 &1  & 0 \\
\hline
0 & 0 & 0 & 0
\end{array} \right ),
\end{equation}
which is already in the no-signaling subspace and of the form $\overline{\gamma}_\text{CH}$. The two expressions are thus recognized as equivalent, with coefficients $\overline{\c}_{abxy} = (-1)^{a+b+x(y+1)}$.

\subsection{Relabelings}
\label{sec:cure:relabelings}
Given coefficients $\c_{abxy}$ for an arbitrary oriented Bell expression (for instance as obtained after the projection and renormalization presented above), equivalent representatives can be obtained by relabeling parties, settings or outcomes. For a finite number of parties, settings and outcomes, the number of possible relabelings is finite as well. 
We notice that relabelings permute the coefficients of this vector, defining an orbit in the space $\Omega^*$, the size of this orbit being finite. We choose to select a canonical representative from this orbit by using lexicographic ordering.

To do so, we first define an enumeration of the coefficients $\c_{abxy}$ using a bijection $M:\ (abxy) \longleftrightarrow i $, with $i=1...D = (m k)^n$, i.e. $\vec{\c}_i=M_i^{abxy} \c_{abxy}$. Let $\vec{\c}'$ and $\vec{\c}''$ be two relabelings of the Bell expression $\vec{\c}$; we order them lexicographically by defining:

\begin{equation}
\vec{\c}' <_\text{lex} \vec{\c}'' \Leftrightarrow \exists j \text{ s.t. } \forall i < j, \vec{\c}'_i = \vec{\c}''_i \text{ and } \vec{\c}'_j < \vec{\c}''_j.
\end{equation}

The \textit{minimal} representative is then the first one under the lexicographic ordering $<_\text{lex}$. As the definition depends on the enumeration of coefficients, we prescribe the following: for a quadruplet $(abxy)$, the next element is found by incrementing first Alice's outcome $a$, then by incrementing Alice's setting $x$, then by incrementing Bob's outcome $b$, finally by incrementing Bob's setting $y$ (as seen in Figure~\ref{fig:coefficientsorder}). This defines the bijection $M$ uniquely.
\begin{figure}
\includegraphics{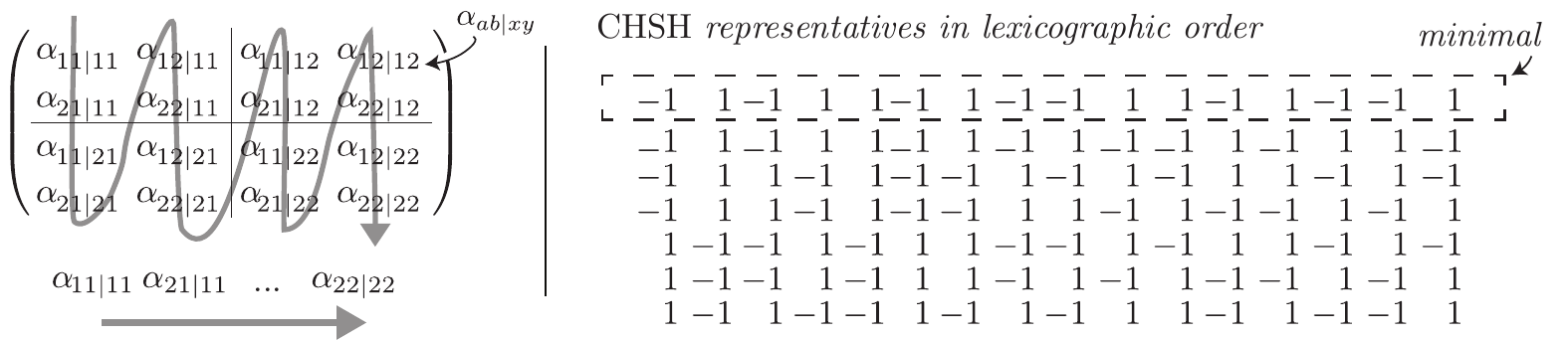}
\caption{Graphical description of the ``column-major'' order of coefficients used in the enumeration for lexicographic ordering, along with a list of the eight representatives of CHSH under relabelings sorted in the lexicographic order. }
\label{fig:coefficientsorder}
\end{figure}
As seen in Appendix~\ref{app:group:faster}, grouping the outcomes and settings of each party individually in the enumeration enables the construction of fast algorithms to select the minimal lexicographic representative. We choose to increment Alice indices first to be compatible with the column-major order used to store multi-dimensional arrays in e.g. MATLAB.

Having defined this order, a particular member of the orbit can be identified either by the permutation that has to be applied to the minimal representative to retrieve the member, or by specifying the rank of this member in the lexicographic order.

Computing the first lexicographic representative of the orbit an inequality belongs to, or computing its rank in the sorted list can be performed, in principle, by enumerating this list completely. However, this approach is not practical for scenarios involving more than a handful number of parties, settings or outcomes. In Appendix~\ref{app:group}, we show how a first representative can be computed quickly by using computational group theory. We also give a fast algorithm to find the index of a representative in the list and obtain the $k^\text{th}$ representative. This allows for the identification of every oriented Bell expression by writing the minimal representative of its orbit, and its rank in the lexicographic order of its orbit.

For the most complex inequalities in~\cite{Avis}, our algorithm can find the first representative by lexicographic order in a list of about $10^{100}$ relabelings in a few seconds on a standard computer.

\subsection{Superfluous parties, settings or outcome distinctions}
\label{sec:cure:superfluous}

As mentioned in~\ref{sec:redundancies:superfluous}, duplicating an outcome, or introducing irrelevant settings or parties produces valid inequalities in scenarios that are larger than strictly needed to express the constraint at hand. In the case of superfluous settings or outcomes distinctions, we call such inequalities \textit{i/o-lifted}. To avoid multiple definitions of equivalent Bell-like inequalities, we write every Bell expression in the smallest scenario in which it is relevant, i.e. such that it is \textit{non-i/o-lifted}.

Note that the coefficients corresponding to irrelevant settings or outcomes playing the same role might not vanish in an expression involving full probabilities. For instance, the positivity $P(1|1) \geq 0$ expressed with one irrelevant input and a duplicate outcome then reads:
\begin{equation}
3P(1|1)-P(2|1)+P(1|2)+P(2|2) \geq 0,
\end{equation}
where all inputs and outcomes intervene. However, they are easily spotted by checking conditions given in section~\ref{sec:redundancies:superfluous}, or by examination of invariance of the inequality under permutation of settings or outcomes. The inequality can then be rewritten in a simplified scenario.

\subsection{Composite expressions}
\label{sec:cure:composite}
Following observation~\eqref{eq:prodExpression}, we say that a Bell expression $\c_{abxy}$ is \textit{composite} if it can be written in a tensor form, i.e. if $\exists$ $\kappa$, $\c'_{ax}$, $\c''_{by}$ such that
\begin{equation}\label{eq:decomp1}
\c_{abxy} =  \kappa \mu_{ax}\mu_{by} + \c'_{ax}\cdot\c''_{by}.
\end{equation}
Otherwise, we call the expression \textit{non-composite}. Here $\mu_{ax}$ and $\mu_{by}$ identify constant terms (c.f. section \ref{sec:cure:no-signaling:basis}).

This decomposition generalizes straightforwardly to cases involving more than one party on each side, and can result in decompositions of the form
\begin{equation}\label{eq:decomp2}
\c_{abcxyz}=\kappa_1 \mu_{ax}\mu_{by}\mu_{cz} + \c'_{ax}(\kappa_2 \mu_{by}\mu_{cz} + \c''_{by}\c'''_{cz}),
\end{equation}
with constants $\kappa_1$ and $\kappa_2$. When $\alpha'$, $\alpha''$ and $\alpha'''$ are non-composite, we refer to them as being the \textit{components} of $\alpha$. We show in Appendix~\ref{compositeAppendix} that when the Bell expressions $\c',\c'',\c'''$ only have rational coefficients, such decomposition is unique up to the sign of the individual expressions. This guarantees that this decomposition can be found by using recursively the following method:
\begin{enumerate}
\item Given a Bell expression $\c$, examine for each separation of the parties into two groups of parties whether a constant $\kappa$ can be added to the expression in order to allow the expression to be written as a tensor product across this separation.
\item When such a biseparation is found, repeat Step 1 for each new Bell expression found. If such biseparation does not exist, the expression is non-composite.
\end{enumerate}

Note that whether a Bell expression is composite or not does not depend on the considered bound, or on the side on which one wishes to bound it. It is thus a property of the expression. Still, in some instances, bounds on an expression can be transmitted through the composition or decomposition operation. We described in section~\ref{sec:redundancies:composite} the condition under which this is possible for composition (proof in Appendix~\ref{compositeAppendix}). Appendix~\ref{compositeAppendix} also describes cases in which a bound $\b$ on an expression like~\eqref{eq:decomp1} can be transmitted to one of its components.

\section{Decomposing Bell-like inequalities in terms of canonical oriented Bell expressions}
\label{sec:decomposing}

In the previous section we showed that any of the following degeneracies in Bell-like inequalities can be dealt with:
\begin{itemize}
\item the orientation ($B(P) \le \b \Leftrightarrow B'(P) \ge -\b$) is fixed by choosing oriented Bell expressions bounded from above,
\item the degeneracy given by the no-signaling constraints is dealt with using the parametrization of Section~\ref{sec:cure:no-signaling:basis},
\item in non-homogenous scenarios, parties and measurements settings are ordered as described in Appendix~\ref{app:nonhomo:canonicalscenarios},
\item the arbitrary constant present in the Bell expression because of the normalization of probability distributions is extracted, and the expression is multiplied by a non-negative factor such that it can be written down using relatively prime integers, as described in Eq.~\eqref{eq:afterprojection},
\item the degeneracy due to relabellings is lifted by looking for the minimal lexicographic representative of the inequality, as described in Section~\ref{sec:cure:relabelings}.
\end{itemize}

Used in this order, each degeneracy removal operation needs only to be applied once. The two operations below can require some of the above operations to be repeated, but only need to be applied a finite number of times:
\begin{itemize}
\item superfluous outcomes or measurements settings are removed by transposing the inequality into a simpler scenario, as described in Section~\ref{sec:cure:superfluous},
\item composite expressions are decomposed as described in Section~\ref{sec:cure:composite}.
\end{itemize}

We say that a non-composite oriented Bell expression is in {\em{canonical form}} when no redundancies are left. Furthermore, we say that two non-composite oriented Bell expressions are {\em{equivalent}} if one can be obtained from the other one by using the transformations above. Any non-composite oriented Bell expression with rational coefficients can thus be described by its canonical form and a transformation associated to each level of degeneracy.

This description also applies to every oriented component of a composite Bell expression. Any composite Bell-like inequality can thus be {\em decomposed} as a combination of the form~\eqref{eq:decomp1} of non-composite oriented Bell expressions, which have, each, one canonical form. Thus, any oriented Bell expression involved in a given Bell-like inequality can be identified in this way. As discussed in Section~\ref{sec:cure:composite}, this combination is in general not unique because of the freedom left in the choice of sign for its components. However, the two orientations of a Bell expression can sometimes be equivalent to each other thanks to some of the degeneracied described above. When this is the case for all components of a composite Bell-like inequality, its decomposition is unique. This is the case for the example given in Figure~\ref{fig:sliwa4}.


\begin{figure}
\includegraphics{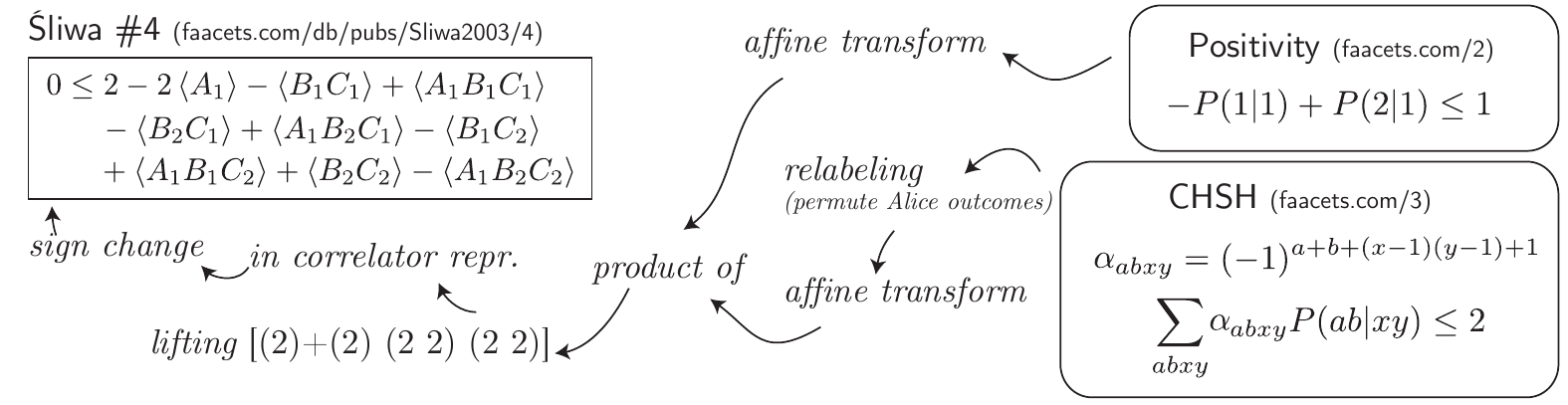}
\caption{Recomposition of a Bell inequality from its decomposition expressed in terms of canonical Bell inequalities, here applied to Sliwa's 4$^\text{th}$ inequality~\cite{Sliwa}. This inequality is a composition of the positivity inequality and of a representative of the CHSH inequality. The URLs link to the entries of the online database.}
\label{fig:sliwa4}
\end{figure}


%
%

\section{A library of oriented Bell expressions}
\label{sec:compendium}

In this paper, we have constructed mathematical and algorithmic tools to deal with all the degeneracies presented. Given that the number of Bell-like inequalities has grown considerably over the last decade, we devised to set up a platform, available both online and offline, to collect the information disseminated over the years in the literature. This platform, named \textit{faacets.com}, contains an implementation of the algorithms described in this work, along with a growing library of inequalities published in peer-reviewed literature. The platform is able to perform decompositions of the kind shown in Figure~\ref{fig:sliwa4} for any Bell expression or Bell-like inequality expressed in terms of rational coefficients. It can also check whether such expression or inequality involves any already-known Bell inequalities or expressions.

After referencing, the platform provides a unique identifier for the canonical form of any non-composite oriented Bell expression. This gives researchers the opportunity to cite a Bell expression by its URL \textit{faacets.com/number}, with the page itself cross-referencing published papers about the expression.

The objects registered in the library are non-composite oriented Bell expressions. They are stored with their properties (which can include a local bound, quantum bound, etc.). Bell expressions which are not invariant under change of orientation, i.e. whose upper and lower bounds might have different properties, can be present twice in the library (once for the lower bound, once for the upper one). This is not the case of CHSH, for instance, which is invariant under change of orientation.

As the number of inequalities is growing, researchers have been providing electronic versions of their results~\cite{Pitowsky,Avis,Vertesi,BancalSym}. To facilitate these exchanges, we have created a human-readable text interchange format for describing Bell inequalities based on YAML~\cite{Yaml}. While the latest specification of this format can be found online~\cite{specfaacets}, an example of a data file can be seen in Figure~\ref{fig:chshyaml}. This format can be edited by hand, parsed using standard YAML tools, and is understood by the software we provide.

\begin{figure}
\includegraphics{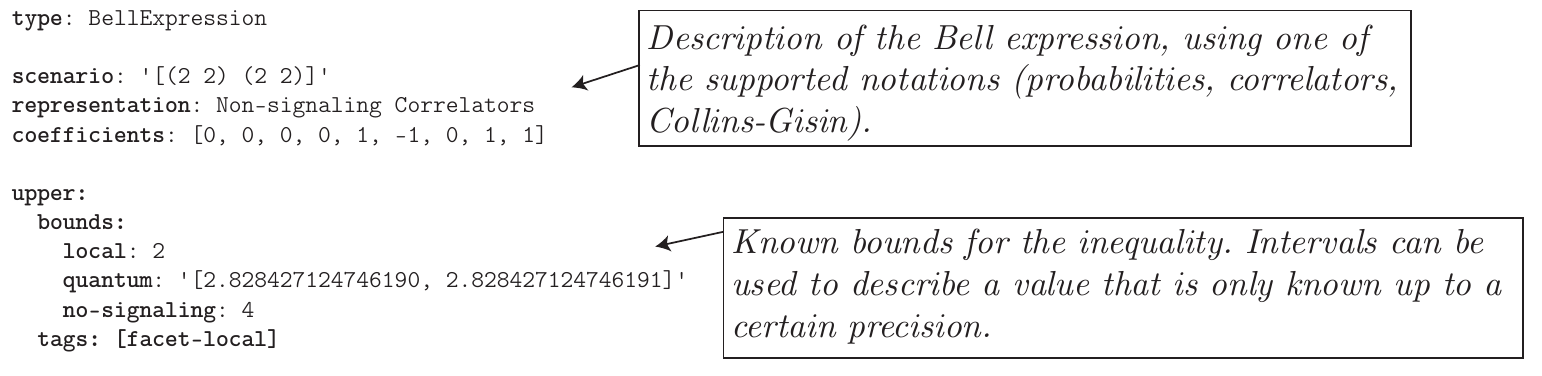}
\caption{Example of a minimal YAML file that describes an oriented Bell expression along with some of its known properties, such as the inequalities it defines. The example here is CHSH. The full specification of the format is available online~\cite{specfaacets}.}
\label{fig:chshyaml}
\end{figure}

The platform can be used either by accessing the website at the URL \textit{faacets.com}, where referenced inequalities can be consulted either in their published or canonical forms, or by downloading the software library~\cite{codefaacets} along with its data files for offline use. The software library is written for the Java platform, and works out-of-the-box from e.g. MATLAB or Python without compiling any package.

The source code is available online~\cite{codefaacets}, and is placed under an open source license. The data files are placed under a Creative Commons license. Both the code and data are developed using public Git repositories, to facilitate open collaboration. We strongly encourage contributions which can either expand the referenced Bell expressions, add information about an inequality already present in the library, or improve and add functionality to the software.

\section*{Conclusion} 
\iffancy\addcontentsline{toc}{section}{\hspace*{-\tocsep}Conclusion}\fi 
We first clarified in Section~\ref{sec:bell-stuff} the concept of Bell inequalities and introduced \textit{oriented Bell expressions} whose bounds define inequalities for different sets of interest (local, quantum, dimension-constrained ...). We then presented in Section~\ref{sec:redundancies} a unified description of several degeneracies present in Bell inequalities. In Section~\ref{sec:cure}, we presented methods to remove each of these redundancies. While doing so, we extended the notion of \textit{liftings}~\cite{liftings} to composite inequalities which involve the product of several Bell expressions and introduced the usage of Computational Group Theory to deal with the relabelings of Bell expressions. In Section~\ref{sec:decomposing}, we showed that these methods can be applied in a consistent and efficient way, to decompose a given Bell inequality in its canonical form. As a consequence of this work, Bell inequalities published in the literature can be classified according to their canonical forms. The online library of inequalities, the open-source software library and the standard file format for inequality interchange are described in Section~\ref{sec:compendium}.

We now encourage contributions to the collection of inequalities contained in the library, either with new Bell expressions, or with additional properties for the ones already referenced. We also encourage researchers to reference Bell inequalities by their canonical index in publications, to facilitate the cross-referencing of research results.

\section*{Acknowledgements} 
\iffancy\addcontentsline{toc}{section}{\hspace*{-\tocsep}Acknowledgements}\fi 
We thank T. Barnea, Y.C. Liang, S. Pironio, G. Puetz and V. Scarani for valuable discussions. This work was supported by the Swiss NCCR “Quantum Science and Technology”, the CHIST-ERA DIQIP, the European ERC-AG QORE, the European SIQS, the Singapore Ministry of Education (partly through the Academic Research Fund Tier 3 MOE2012-T3-1-009) and the Singapore National Research Foundation.

\appendix
\section{Non-homogeneous scenarios}
\label{app:nonhomogeneous}
For pedagogical reasons, the main text focuses on the elimination of redundancies for \textit{homogeneous} scenarios, where the number of measurement settings $m$ and outcomes $k$ is the same everywhere. But, as we will see below, the non-degenerate form of some inequalities can only be provided using non-homogeneous scenarios. We introduce the following notation for those scenarios: $[(k_{11} k_{12} \ldots k_{1 m_1})~(k_{21} k_{22} \ldots k_{2 m_2}) \ldots (k_{n 1} k_{n 2} \ldots k_{n m_n})]$, where $m_i \ge 1$ is the number of measurement settings for the $i^\text{th}$ party and $k_{i j} \ge 2$ is the number of measurement outcomes for the $j^\text{th}$ measurement setting of the $i^\text{th}$ party.

One of the earliest examples of non-degenerate Bell inequality in a non-homogeneous scenario is the one given in~\cite{Pironio14}, in the scenario $[(2~3)~(2~2~2)]$ given below in the Collins-Gisin notation:
\begin{eqnarray}
\label{eq:pironiononhomogeneous}
P_\text{A}(1|1) + P_\text{B}(1|1) + P_\text{B}(1|2) - P_\text{AB}(11|11) - P_\text{AB}(11|12) - P_\text{AB}(11|21) - P_\text{AB}(21|22) & & \nonumber \\
- P_\text{AB}(11|13) + P_\text{AB}(11|23) + P_\text{AB}(21|23) & \ge & 0.
\end{eqnarray}
The simplest homogeneous scenario in which a lifted version of the inequality~\eqref{eq:pironiononhomogeneous} appears is $(2,3,3)$ or $[(3~3~3)~(3~3~3)]$. But in our classification scheme, the canonical form of this inequality resides in a non-homogeneous scenario.

\subsection{Canonical scenarios and relabelings}
\label{app:nonhomo:canonicalscenarios}

In the non-homogeneous scenarios $[(2~3)~(2~2~2)]$, the exchange of the settings of Alice is not a relabeling because it affects the Bell scenario: Alice's settings have a different number of outcomes. For the same reason, Alice and Bob cannot be permuted. Thus, non-homogenous scenarios have intrinsic restrictions on the relabelings that can be performed.

Still, the non-homogeneous scenarios $[(2~3)~(2~2~2)]$ and $[(3~2)~(2~2~2)]$ describe essentially the same physical system, and inequalities can be transfered from one scenario to the other one by applying a \textit{reordering} of the parties and measurement settings (not to be confused with relabellings which do not affect the Bell scenario). Thus, before looking for the minimal lexicographic representative of an expression, we have to reorder the expression so that its scenario is in a \textit{canonical form}.

\begin{definition}
A scenario is in the \textit{canonical form} if:
\begin{itemize}
\item for successive settings, the number of measurement outcomes is nonincreasing: $k_{i j} \ge k_{i, j+1}$,
\item parties are ordered lexicographically: for all successive non-identical parties $i$ and $i+1$, there is a $j \ge 0$ such that $\forall k < j$ we have $k_{i k} = k_{i+1, k}$ and $k_{i j} > k_{i+1, j}$. For these ordering purposes, we define $k_{i j} = 0$ for $j > m_i$.
\end{itemize}
\end{definition}

The canonical form of a given scenario can always be found by reordering parties and settings: the canonical form corresponding to $[(2~3)~(2~2~2)]$ is $[(3~2)~(2~2~2)]$, and the inequality~\eqref{eq:pironiononhomogeneous} becomes:

\begin{eqnarray}
\label{eq:pironiocanonicalscenario}
-P_\text{A}(1|2) - P_\text{B}(1|1) - P_\text{B}(1|2) + P_\text{AB}(11|21) + P_\text{AB}(11|22) + P_\text{AB}(11|11) + P_\text{AB}(21|12) & & \nonumber \\
+ P_\text{AB}(11|23) - P_\text{AB}(11|13) - P_\text{AB}(21|13) & \le & 0,
\end{eqnarray}
where we have changed the lower bound to an upper bound.

\subsection{Parametrizing the no-signaling subspace}

Having reordered the scenario, we now construct a basis for the vector space of Bell-like expressions in the new scenario. In any non-homogenous scenario, the construction given in Section~\ref{sec:cure:no-signaling:basis} is still correct, if we choose the range of all indices to cover all existing outcomes and settings. We give below an example of this parametrization for the expression~\eqref{eq:pironiocanonicalscenario}. Enumerating Alice indices $(a|x)$ as $(1|1), (2|1), (3|1), (1|2), (2|2)$, we construct the following basis:
\begin{equation}
\vec{\mu}_\text{A} = \left ( \begin{array}{r} 1/2 \\ 1/2 \\ 1/2 \\ \hline 1/2 \\ 1/2 \end{array} \right ),
\vec{\lambda}^{1 1}_\text{A} = \left ( \begin{array}{r} 1 \\ -1 \\ 0 \\ \hline 0 \\ 0 \end{array} \right ),
\vec{\lambda}^{2 1}_\text{A} = \left ( \begin{array}{r} 0 \\ 1 \\ -1 \\ \hline 0 \\ 0 \end{array} \right ),
\vec{\lambda}^{1 2}_\text{A} = \left ( \begin{array}{r} 0 \\ 0 \\ 0 \\ \hline 1 \\ -1 \end{array} \right ),
\vec{\nu}_\text{A} = \left ( \begin{array}{r} 1 \\ 1 \\ 1 \\ \hline -1 \\ -1 \end{array} \right ),
\end{equation}
and for Bob, with the enumeration $(b|y)$ as $(1|1), (2|1), (1|2), (2|2), (1|3), (2|3)$:
\begin{equation}
\vec{\mu}_\text{B} = \left ( \begin{array}{r} 1/3 \\ 1/3 \\  \hline 1/3 \\ 1/3 \\ \hline 1/3 \\ 1/3 \end{array} \right ),
\vec{\lambda}^1_\text{B} = \left ( \begin{array}{r} 1 \\ -1 \\ \hline 0 \\ 0 \\ \hline 0 \\ 0 \end{array} \right ),
\vec{\lambda}^2_\text{B} = \left ( \begin{array}{r} 0 \\ 0 \\ \hline 1 \\ -1 \\  \hline 0 \\ 0 \end{array} \right ),
\vec{\lambda}^3_\text{B} = \left ( \begin{array}{r} 0 \\ 0 \\ \hline 0 \\ 0 \\  \hline 1 \\ -1 \end{array} \right ),
\vec{\nu}^1_\text{B} = \left ( \begin{array}{r} 1 \\ 1 \\ \hline -1 \\ -1 \\  \hline 0 \\ 0 \end{array} \right ),
\vec{\nu}^2_\text{B} = \left ( \begin{array}{r} 0 \\ 0 \\ \hline 1 \\ 1 \\  \hline -1 \\ -1 \end{array} \right ).
\end{equation}
Writing $P_\text{A}(1|2) = P_\text{AB}(11|21) + P_\text{AB}(12|21)$, $P_\text{B}(1|1) = P_\text{AB}(11|11) + P_\text{AB}(21|11) + P_\text{AB}(31|11) $, and $P_\text{B}(1|2)= P_\text{AB}(11|12) + P_\text{AB}(21|12) + P_\text{AB}(31|12)$, the table for the original coefficients $\c_{ab|xy}$ look like:
\begin{equation}
\left(
\begin{array}{rr|rr|rr} 
\c_{11|11}  &  \c_{12|11}  & \c_{11|12}  &  \c_{12|12}  &  \c_{11|13}  &  \c_{12|13}  \\  
\c_{21|11}  &  \c_{22|11}  & \c_{21|12}  &  \c_{22|12}  &  \c_{21|13}  &  \c_{22|13}  \\  
\c_{31|11}  &  \c_{32|11}  & \c_{31|12}  &  \c_{32|12}  &  \c_{31|13}  &  \c_{32|13}  \\  
\hline
\c_{11|21}  &  \c_{12|21}  & \c_{11|22}  &  \c_{12|22}  &  \c_{11|23}  &  \c_{12|23}  \\  
\c_{21|21}  &  \c_{22|21}  & \c_{21|22}  &  \c_{22|22}  &  \c_{21|23}  &  \c_{22|23}
\end{array}
\right)
=
\left(
\begin{array}{rr|rr|rr} 0  &  0  &  -1  &  0  &  -1  &  0  \\  -1  &  0  &  0  &  0  &  -1  &  0  \\  -1  &  0  &  -1  &  0  &  0  &  0  \\ \hline 0  &  -1  &  1  &  0  &  1  &  0  \\  0  &  0  &  0  &  0  &  0  &  0 \end{array}
\right) \le 0,
\end{equation}
After decomposition in the basis above, the coefficient table for $\gamma$ is, before and after projection:
\begin{equation}
\left (
\begin{array}{r|rrr|rr}   &  -2  &  -2  &  -2  &  -8  &  -4  \\ \hline 0  &  8  &  -4  &  -4  &  8  &  4  \\  0  &  4  &  4  &  -8  &  4  &  8  \\  6  &  6  &  6  &  6  &  -8  &  -4  \\ \hline  -15  &  -7  &  -7  &  -7  &  4  &  2 \end{array}
\right ) \le 18
 \quad \overset{\text{projection}}{\xrightarrow{\hspace*{1.5cm}}} \quad
\left (
\begin{array}{r|rrr|rr}   &  -1  &  -1  &  -1  &  0  &  0  \\ \hline  0  &  4  &  -2  &  -2  &  0  &  0  \\   0  &  2  &  2  &  -4  &  0  &  0  \\  3  &  3  &  3  &  3  &  0  &  0  \\ \hline   0  &  0  &  0  &  0  &  0  &  0 \end{array}
\right ) \le 9.
\end{equation}
Which converted back to coefficients of the type $\c_{abxy}$ gives, before and after search of the minimal lexicographic representative:
\begin{equation}
\left(
\begin{array}{rr|rr|rr}  7  &  -7  &  -5  &  5  &  -5  &  5  \\   -5  &  5  &  7  &  -7  &  -5  &  5  \\  -5  &  5  &  -5  &  5  &  7  &  -7  \\ \hline   7  &  -3  &  7  &  -3  &  7  &  -3  \\  -9  &  5  &  -9  &  5  &  -9  &  5 \end{array}\right) \le 18
 \quad \overset{\text{minimal lex. repr.}}{\xrightarrow{\hspace*{2.5cm}}} \quad
\left(
\begin{array}{rr|rr|rr}  -7  &  7  &  -5  &  5  &  -5  &  5  \\  5  &  -5  &  -5  &  5  &  7  &  -7  \\  5  &  -5  &  7  &  -7  &  -5  &  5  \\ \hline  -3  &  7  &  7  &  -3  &  7  &  -3   \\  5  &  -9  &  -9  &  5  &  -9  &  5 \end{array}
\right) \le 18
\end{equation}

\section{Computational group theory and Bell scenarios}
\label{app:group}
We define $G$ as the group of all relabelings of parties, settings and outcomes which are compatible with a scenario. When the scenario is homogeneous, the order of that group is $|G| = n! (m!)^n (k!)^{nm}$, corresponding to $n!$ relabelings of parties, $m!$ relabelings for the settings of $n$ parties, $k!$ relabelings of outcomes for $m$ settings of $n$ parties.

An explicit construction of this group (notation, multiplication rules) will be provided in future work; for the purposes of this Appendix, we only need the observation that this group can be generated by the relabelings of each pair of adjacent parties, settings or outcomes - for a homogeneous scenario, this represents $(n-1) + n(m-1) + n m(k-1)$ generators.

After projection in the no-signaling subspace, any Bell expression is represented by a vector of coefficients $\vec{c}\in \mathbb{R}^D$ with coefficients $c(i)$, where each index $i$ represents a tuple $(a,b,...,x,y,...)$. The relabeling of parties, settings or outcomes can be represented by a permutation $\pi$ acting on the indices $(a,b,...,x,y,...)$ and thus on $i$. The action of $\pi$ on $\vec{\c}$ is then defined by $\c^\pi(i^\pi) = \c(i)$, and this action is faithful (i.e. only the action of the identity element of $G$ leaves all $\vec{\c}$ invariant). We also use the {\em{right}} action in this Appendix, i.e. for $g,h \in G$, $i^{gh} = (i^g)^h$.

The study of permutation groups using bases and strong generating sets is exposed at length in~\cite{Holt}, from which we extract the following key points: we study $G$ together with its permutation action on vectors $\vec{\c}$, and define $G^{[1]}$ the subgroup of $G$ of order $|G^{[1]}|$ that leaves the first element of $\vec{\c}$ invariant: $\forall g_1 \in G^{[1]}, 1^{g_1} = 1$. By Lagrange's Theorem, there is a set $U_1$ of $n_1 = \frac{|G|}{|G^{[1]}|}$ elements of $G$ such that every $g\in G$ can be written as:
\begin{equation}
g = g_1 u_1, \qquad u_1 \in U_1, g_1 \in G^{[1]}.
\end{equation}
Moreover, every $u_1 \in U_1$ has a different image $1^{u_1}$. This decomposition can be iterated, by writing $G^{[2]}$ the subgroup of $G^{[1]}$ that leaves the second element of $\vec{c}$ invariant: $\forall g_2 \in G^{[2]}, 2^{g_2} = 2$, and there exists a set $U_2$ of $n_2 = \frac{|G^{[1]}|}{|G^{[2]}|} $ elements of $G^{[1]}$ such that every $g_1 \in G^{[1]}$ can be written as $g_1 = g_2 u_2$ with $ u_2 \in U_2, g_2 \in G^{[2]}$, and every $u_2\in U_2$ having a different image $2^{u_2}$. The last group in the iteration, $G^{[D-1]}$, stabilizing the first $D-1$ elements of $\vec{c}$ is the trival group containing only the identity.  Then, every element of $g\in G$ can be decomposed as:
\begin{equation}
\label{eq:gdecomposition}
g = u_{D-1} ... u_2 u_1, \qquad u_j \in U_j.
\end{equation}

Such a decomposition of $G$ is known as a stabilizer chain, and can be computed efficiently using the Schreier-Sims algorithm, with the prescribed base $1,2,...D-1$. To do so, the faster randomized version can be used without reservations because the order of the group $|G|$ is known in advance. Then, any permutation of $\vec{\c}$ can be decomposed using~\eqref{eq:gdecomposition}:
\begin{equation}
\c^{u_1^{-1} ... u_{d-1}^{-1}}(i) = \c(i^{u_{d-1} ... u_1}),
\end{equation}
and our algorithms will then be based on the observation that the first coefficient $\c^{g^{-1}}(1)$ is selected by $u_1$ only, because $1^{u_{D-1} ... u_2} = 1$. Stabilizer chains also enable the fast computation of subgroups and their order~\cite{Holt}.

We give a sketch below of three algorithms. The first two are used to compute the minimal lexicographic representative of a Bell expression, while the third one can compute the lexicographic rank of a Bell expression, or retrieve a particular representative by its rank. While a complexity analysis of these algorithms is outside the scope of the present work, these algorithms perform well enough to find the minimal lexicographic representative search for any Bell expression given in~\cite{Avis,BancalSym} in seconds on a standard computer.

\subsection{Algorithm to find the minimal representative by lexicographic order}
\label{app:group:minimal}
\begin{figure}
\includegraphics{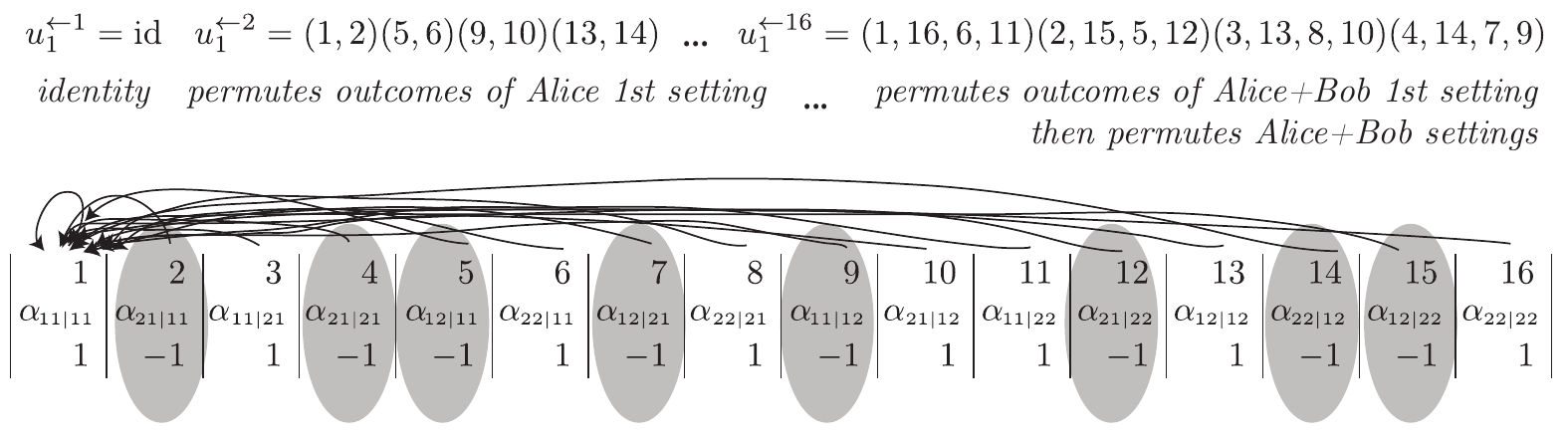}
\caption{Illustration of the first step of Algorithm~\ref{alg:minimal}, which select appropriate candidates based on their action on the first coefficient of the Bell expression, applied to the coefficients of the CHSH inequality~\eqref{eq:chshprojected}. The transversal elements $u_1^{\leftarrow ...}$ corresponding to gray elements are selected as candidates for the next step. The permutations $u_1^{\leftarrow 2}, u_1^{\leftarrow 16}$ are written using the cycle notation, along with a description of the relabeling. Here, the first coefficient can be taken from anywhere, because there is a relabeling between any tuple of indices $(ab|xy)$ and $(11|11)$.  }
\label{fig:algo1}
\end{figure}

\begin{algorithm}
\begin{algorithmic}                    
  \REQUIRE \\
  The current coefficient index $j$ to explore. \\
  The set of candidates $H^j \subset G$.
  \ENSURE \\
  The new set of candidates $H^{j+1}$.
  \\
  \STATE $m  \leftarrow \infty $, $ H^{j+1} \leftarrow \emptyset $
  \FOR{$u \in U^j, h \in H^j $}
  \IF{$ \c(j^{u h}) < m $ }
  \STATE $ m \leftarrow \c(j^{u h}) $, $ H^{j+1} \leftarrow \emptyset $
  \ENDIF
  \IF{$ \c(j^{u h}) = m $ }
  \STATE $ H^{j+1} \leftarrow H^{j+1} \cup \{ u h \} $
  \ENDIF
  \ENDFOR
\end{algorithmic}
\caption{Algorithm to filter appropriate permutation candidates by their action on the coefficient index $j$. Iterated usage of this algorithm for $j=1 ... D-1$ will find permutations of a given current Bell expression to its minimal lexicographic representative.}
\label{alg:minimal}
\end{algorithm}

Given a vector $\vec{\c}$, we want to find a (possibly not unique) element $g\in G$ such that $\vec{\c}^{~g^{-1}}$ is lexicographically minimal. To do so, we use the decomposition~\eqref{eq:gdecomposition}, and look for the candidates such that $\c^{g^{-1}}(1) = \c(1^g)$ is minimal, as shown in Figure~\ref{fig:algo1}. Because $1^{u_j}=1$ for $j > 1$, we only need to filter the $u_1$ for which $\c(1^{u_1})$ is minimal. This procedure can then be repeated for all $j = 1 ... d - 1$. To do so, we start with the index $j = 1$ and the set of permutation candidates $H^1 = \{ \text{identity} \} $, and call Algorithm~\ref{alg:minimal} for $j = 1 ... D-1$. Any permutation from the final set $h \in H^D$ will give us a $\vec{\c}^{~h^{-1}}$ lexicographically minimal.

\subsection{Faster variant of minimal lexicographic representative algorithm}
\label{app:group:faster}
Building on this first algorithm, we construct a faster variant by enumerating the elements of $\c$ using two indices $i$ and $j$, with the index $i$ corresponding to $(a,x)$ and the index $j$ corresponding to the remaining $(b,c,...,y,z,...)$. With this partitioning of indices, the first party is singled out and cannot be permuted with another party, and thus we restrict ourselves to relabelings of settings and outcomes, and relabelings of parties except the first. The algorithm has then to be run several times for each possible first party, and from there the overall minimal representative is chosen. The allowed permutations can be expressed using elements of two permutation groups, $h \in H$ acting on the index $i$ and $g \in G$ acting on the index $j$, such that:
\begin{equation}
\label{eq:matrixpermutation}
\c(i,j)^{(h,g)} = \c(i^{h^{-1}},j^{g^{-1}}).
\end{equation}

As shown in Figure~\ref{fig:algo2}, the object $\c(i,j)$ can be viewed as a matrix, with $H$ acting on the $I$ rows and $G$ acting on the $J$ columns of the matrix.. We write the columns of the matrix $\c(i,j)$ as vectors $\vec{v}_j = (\c(1,j), ... \c(I,j) )$. Note that when the lexicographic order is used to compare Bell expressions, this comparison is done column by column on the matrices $\c(i,j)$.
\begin{figure}[b]
\includegraphics{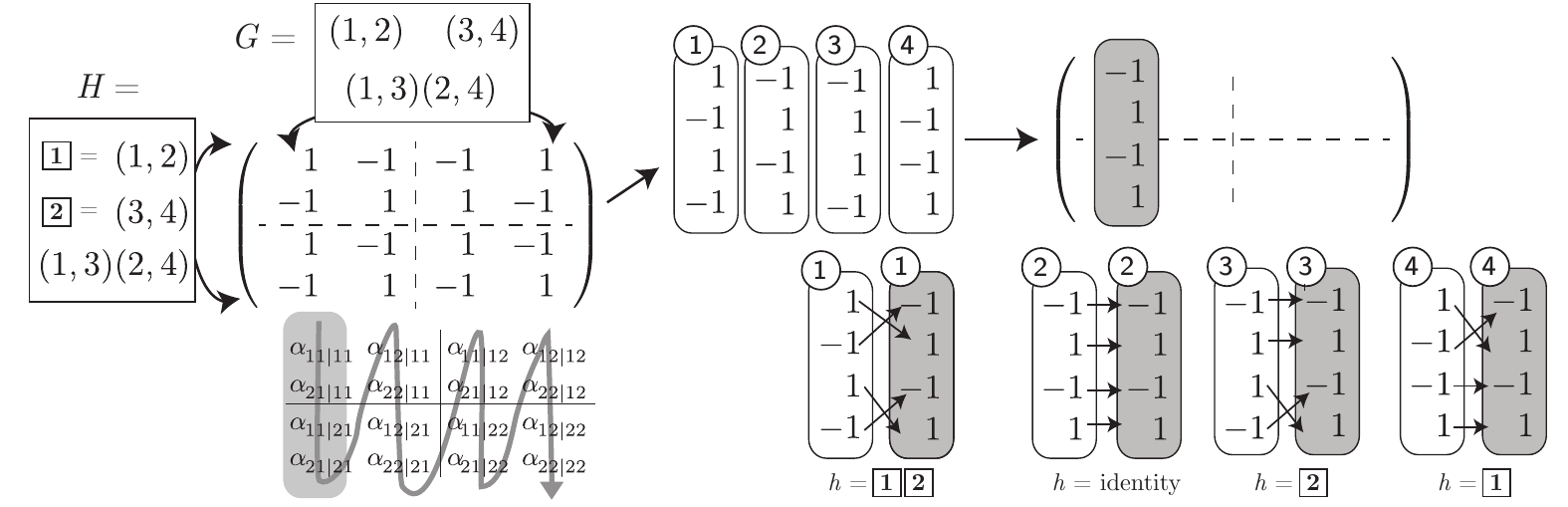}
\caption{Illustration of Algorithm~\ref{alg:matrixfirstlexico}, an optimized variant of Algorithm~\ref{alg:minimal} obtained by partitioning the indices of the first and subsequent parties, and representing the Bell expression as a matrix, whose rows and columns can be permuted. The example shown here is the CHSH inequality. The permutation groups $H$ and $G$ for the rows and columns are shown along with their generators, who correspond to the permutation of outcomes for the first and second measurement setting, and the permutation of measurement settings. Instead of selecting candidates coefficient by coefficient, candidates are filtered column by column. We show here the selection of the candidate for the first column of the coefficient matrix, which can be obtained by placing there a permutation of any of the four original columns, giving four candidates for the next iteration.}
\label{fig:algo2}
\end{figure}

The algorithm is run iteratively for each column, fixing the first $\gamma - 1$ columns to their lexicographic minimal coefficients. Let $\gamma$ be the current column under consideration, and $C_\gamma$ the complete set of current candidate matrices, with the coefficients corresponding to the columns $j < \gamma$ lexicographically minimal. By construction of the algorithm, these candidates are obtained by row permutations and column permutations of the type $u_{\gamma - 1} ... u_1$. To simplify the notation in the algorithm description here, the elements of $C_\gamma$ are the permuted matrices, while in the implementation, the permutated matrices are implicitly described by the permutation pair $(h,g)$ and the original non-permuted matrix.

If $C_\gamma$ is a complete set of candidates having their first $\gamma-1$ columns lexicographically minimal, we can only permute the columns with $j \ge \gamma$ using $G^{[\gamma-1]}$, and apply row permutations that leaves the first $\gamma - 1$ columns invariant. Let $S \subseteq H$ be the maximal subgroup of $H$ that leaves these column vectors invariant.

\begin{algorithm}[t]
\begin{algorithmic}                    
  \REQUIRE \\
  The current column $\gamma$ to put minimal. \\
  The set of candidates $C_\gamma$. \\
  The permutation groups $S$ and $G^{[\gamma-1]}$.

  \ENSURE \\
  The new set of candidates $C_{\gamma+1}$.
  \\
  \STATE $\vec{m} \leftarrow (\infty ... \infty)$, $C_{\gamma+1} \leftarrow \emptyset$

  \FOR{$\delta \in C_\gamma, u \in U^\gamma, h \in S$}
  \STATE $\vec{v} \leftarrow \vec{v} \text{ s.t. } v_i = \delta(i^h, \gamma^u)$
  \IF{$\vec{v} < \vec{m}$}
  \STATE $\vec{m} \leftarrow \vec{v}$, $C_{\gamma+1} \leftarrow \emptyset$
  \ENDIF
  \IF{$\vec{v} = \vec{m}$}
  \STATE $C_{\gamma+1} \leftarrow C_{\gamma+1} \cup \{ \delta^{(h^{-1}, u^{-1})} \} $
  \ENDIF
  \ENDFOR
\end{algorithmic}
\caption{Column candidates selection algorithm for the current column $\gamma$ under consideration.}
\label{alg:matrixfirstlexico}
\end{algorithm}

We want now to construct the set $C_{\gamma+1}$ of candidates with their first $\gamma$ columns lexicographically minimal. These candidates can be obtained by applying the best row permutation from $S$ and using a column permutation from $U_\gamma$, as every element of $g \in G^{[\gamma-1]}$ can be decomposed as $g = u_{J-1} ... u_\gamma$, with the permuted column at index $\gamma$ chosen by $u_\gamma$. 

This is achieved using Algorithm~\ref{alg:matrixfirstlexico}, with initial set $C_1 = \{ \c \}$, iterating through $\gamma = 1 ... J-1$. An additional speed-up is given by filtering adequate $h \in S$ using Algorithm~\ref{alg:minimal} for each column selected by $u\in U^\gamma$.

\subsection{Algorithm to determine the $k^\text{th}$ representative and find the rank of a representative}
\label{app:group:rank}

All representatives of a Bell expression under relabelings can be enumerated and sorted in the lexicographic order, as is shown in the left part of Figure~\ref{fig:algo3}. Any representative can be distinguished by computing its rank in the list. We provide here a fast algorithm to compute the rank of a representative, and to retrieve the $k^\text{th}$ lexicographic representative in the list, while computing only a few explicit representatives. Our algorithm is based on the following observation.
\begin{proposition}
\label{prop:numberrepr}
Let $G$ be a permutation group acting on a Bell expression $\vec{\c}$. Let $S \subseteq G$ be the maximal subgroup of $G$ that leaves $\vec{\c}$ invariant, i.e. $\forall s \in S, \vec{\c}^s = \vec{\c}$. Then the number of representatives of $\vec{\c}$ under relabelings by $G$ is given by $n = \frac{|G|}{|S|}$. If $G$ is the group of all relabelings of parties, settings and outcomes, $n$ is the total number of representatives of $\vec{\c}$ under relabelings.
\end{proposition}
\begin{proof}
The right cosets of $S$ in $G$ are written $S g$ for a $g \in G$. All elements of a coset of $S$ in $G$ lead to the same representative of $\vec{\c}$. Indeed, $\forall s \in S$,
\begin{equation}
\vec{\c}^{s g} = (\vec{\c}^s)^g= \vec{\c}^g,
\end{equation}
and because $S$ is the maximal subgroup of $G$ that leaves $\c$ invariant, all permutations leading to the same representative belong to the same coset. The number of right cosets, and thus of representatives is then given by Lagrange's theorem.
\end{proof}

\begin{figure}
\includegraphics{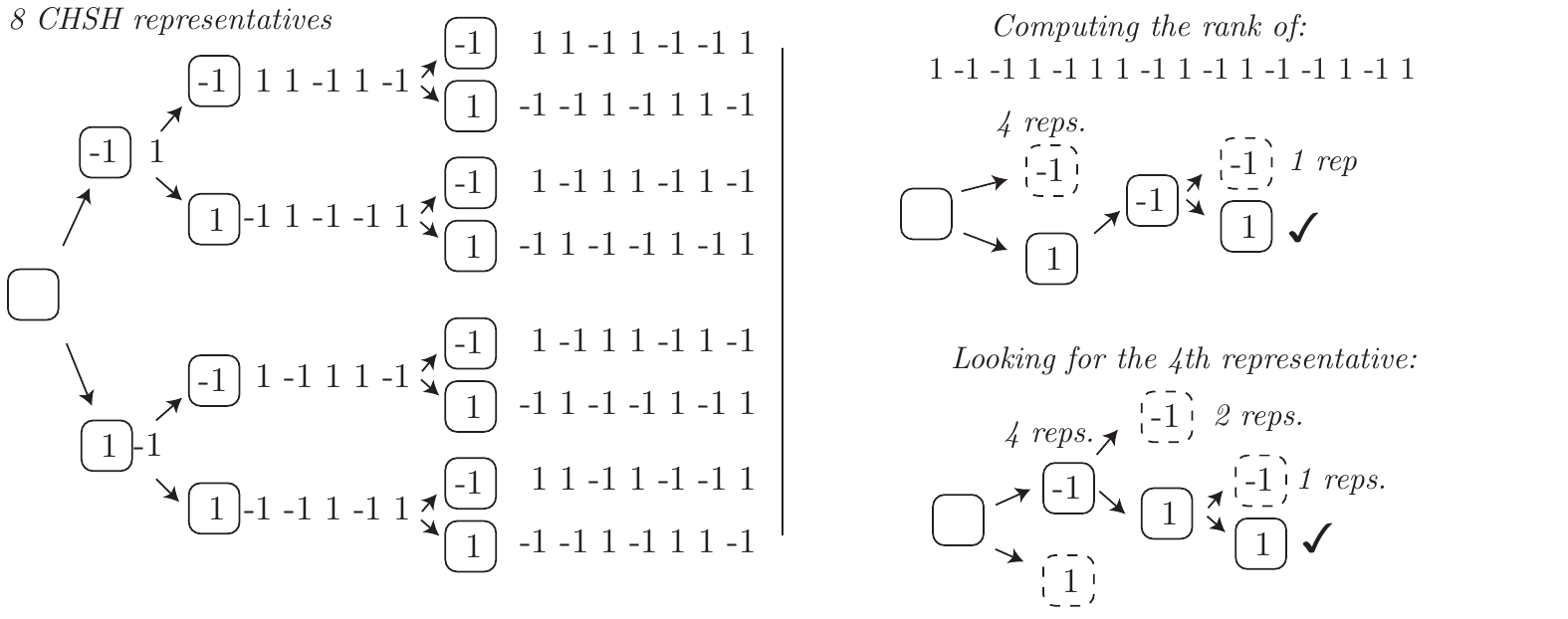}
\caption{List of the 8 relabelings of the CHSH inequality~\eqref{eq:chshprojected} already enumerated in Figure~\ref{fig:coefficientsorder}, shown here as a tree, with their 16 coefficients enumerated from left to right. As shown in the right part, the number of representatives contained in a branch can be computed without exploring the branch at all. Top-right: the index of the considered representative is computed to be $4+1+1=6$. Bottom-right: to find the $4^\text{th}$ representative we select the coefficient $-1$ at the first branch because the $4^\text{th}$ representative is inside this first block of size $4$. At the second branch, we skip the first block of size $2$, and at the third branch then skip a block of size $1$ to obtain the looked-for representative. }
\label{fig:algo3}
\end{figure}

Given a Bell expression $\vec{\c}$ and its permutation group $G$, we associate to each sequence $\vec{f} = (f_1, ... f_j)$ with $j \in \{1,...,D-1\}$ the \textit{block} $B_{\vec{f}}$ of representatives of $\vec{\c}$ under relabelings whose coefficients start with $\vec{f}$:
\begin{equation}
B_{\vec{f}} = \left \{ \vec{\c}^g \middle | g \in G, \forall i \le j, \c^g(i) = f_i \right \}.
\end{equation}
The block $B_{\vec{f}}$ can be described by the action of the subgroup $G^{[j]}$ on a small set of candidates $C_{\vec{f}}$ such that:
\begin{equation}
\label{eq:app:group:rank:B}
B_{\vec{f}} = \left \{ \vec{c}^{~g} \middle | g \in G^{[j]}, \vec{c} \in C_{\vec{f}} \right \}.
\end{equation}

We say that the set $C_{\vec{f}}$ is minimal when all $b \in B_{\vec{f}}$ in~\eqref{eq:app:group:rank:B} are generated by an unique $\vec{c} \in C_{\vec{f}}$, which implies that for all $\vec{c}_1 \ne \vec{c}_2 \in C_{\vec{f}}$, there is no $h \in G^{[j]}$ such that $\vec{c}_1^h = \vec{c}_2$. When the set $C_{\vec{f}}$ is minimal, the number of representatives in $ B_{\vec{f}}$ can be found using Proposition~\ref{prop:numberrepr} on every element in $C_{\vec{f}}$. 

A sequence of coefficients $\vec{f}$ of length $j$ can be extended by adding another coefficient $f_{j+1}$. To each feasible $f_{j+1}$ corresponds a new block $B_{(f_1 ...  f_j f_{j+1})}$. These blocks can be ordered by the value of $f_{j+1}$, and the number of their representatives is easily computed. As shown in Figure~\ref{fig:algo3}, this enables the fast computation of the rank of a given representative or the search of a representative of given rank.

A minimal set $C_{\vec{f}}$ can be constructed by taking all $\vec{c} = \vec{\c}^{~u_1^{-1} ... u_j^{-1} h^{-1}}$ with $ u_k \in U_k$, $h \in G^{[j]}$ such that $ c_i = f_i $ for $i \le j$ and $\vec{c}$ is lexicographically minimal under permutation by $G^{[j]}$.  We show in Algorithm~\ref{alg:block} how to construct such minimal sets $C_{\vec{f}}$ by growing $\vec{f}$ one element at a time, starting with $j = 0$ and $C_{()} = \{ \vec{m} \}$ containing the minimal lexicographic representative $\vec{m}$ of $\vec{\c}$.

\begin{algorithm}
\begin{algorithmic}                    
  \REQUIRE \\
  The current level $j$ to explore, and the starting coefficients $\vec{f}$. \\
  The set of candidates $C_{\vec{f}}$. \\
  The permutation groups $G^{[j]}$ and $G^{[j+1]}$.

  \ENSURE \\
  Sequence of coefficients $F = \{ f_{j+1} \}$ indexing the blocks. \\
  The number $N_{f_{j+1}}$ of representatives in the block $B_{(f_1 ... f_j f_{j+1})}$. \\
 Minimal sets $C_{(f_1 ... f_j f_{j+1})}$.
  \\
  \STATE  $F \leftarrow \left \{ c^{u^{-1}}(j) \middle | u \in U_{j+1}, \vec{c} \in C_{\vec{f}} \right \} $

  \FOR{$f_{j+1} \in F$}
  \STATE $C_{(f_1 ... f_j f_{j+1})} \leftarrow \emptyset $
  \FOR{$u \in U_{j+1}, \vec{c} \in C_{\vec{f}}$}
  \IF{$c^{u^{-1}}(j) = f_{j+1}$}
  \STATE $C_{(f_1 ... f_j f_{j+1})} \leftarrow C_{(f_1 ... f_j f_{j+1})} \cup \left \{ \text{min}_{G^{[j+1]}} \vec{c}^{u^{-1}} \right \}$
  \ENDIF
  \ENDFOR
  \STATE $N_{f_{j+1}} \leftarrow 0$
  \FOR{$c \in C_{(f_1 ... f_j f_{j+1})} $}
  \STATE $N_{f_{j+1}} \leftarrow N_{f_{j+1}} + \frac{\left |G^{[j+1]} \right|}{\left|S_{G^{[j+1]}}(\vec{c})\right|}$
  \ENDFOR
  
  \ENDFOR
\end{algorithmic}
\caption{Block construction algorithm for the $k^\text{th}$ representative algorithm. In the algorithm, we write $\text{min}_{G^{[j+1]}} \vec{c}^{u^{-1}} $ the minimal lexicographic representative of $\vec{c}^{u^{-1}} $ under permutation by the group $G^{[j+1]}$, and $S_{G^{[j+1]}}(\vec{c})$ the maximal subgroup of $G^{[j+1]}$ that leaves $\vec{c}$ invariant. We note that once the elements of $F$ have been computed, the characteristics of each block can be computed in parallel and on demand; when looking for a particular representative, only the blocks with $f_{j+1}$ less or equal to the looked-for coefficient need to be computed.}
\label{alg:block}
\end{algorithm}

\section{Composite Bell expressions}\label{compositeAppendix}
In Section~\ref{sec:cure:composite} of the main text we introduced the notion of composite Bell expressions. In this Appendix we prove the properties of these expressions and of their bounds mentioned in the main text.

\subsection{Identifying composite expressions and their general form}
Let us recall the definition given in Section~\ref{sec:cure:composite}.

\begin{definition} A bipartite Bell expression $\c_{abxy}$ is composite if $\exists$ $\kappa$, $\c'_{ax}$, $\c''_{by}$ such that
\begin{equation}\label{eq:defineComposite}
\c_{abxy} = \kappa \mu_{ax}\mu_{by} + \c'_{ax}\c''_{by}.
\end{equation}
\end{definition}
This definition naturally extends to Bell expressions involving more parties, by letting $\c'$ and $\c''$ describe expressions for two disjoint subset of parties. In this Appendix $(a,x)$ should thus be understood as the settings and outcomes of a group of parties, and similarly for $(b,y)$. We then refer to A and B as two (non-overlaping) groups of parties.

As discussed in the main text, the coefficients $\c_{abxy}$ of an arbitrary Bell expression are not uniquely defined due to the normalization and no-signaling conditions. This opens the possibility for an expression which doesn't satisfy Eq.~\eqref{eq:defineComposite} to be equivalent over the set of no-signaling correlations to another one, $\c'_{abxy}$, which might admit such decomposition. Let us show that this degeneracy can be dealt with through the choice of parametrization provided in section~\ref{sec:cure:no-signaling} because any Bell expression which can be written as composite for no-signaling correlations must appear so in this parametrization.

\begin{proposition}\label{propositionSig}
If a Bell expression $\c_{abxy}$ can be decomposed into two expressions $\c'_{ax}$, $\c''_{by}$, then its parametrization $\overline{\c}_{abxy}$ (described in Eq.~\ref{eq:cbar}) can also be decomposed into two expressions $\overline{\c}'_{ax}$, $\overline{\c}''_{by}$.
\end{proposition}
\begin{proof}
Tensor product structures are independent of particular choices of bases. Let us thus use the basis $\{u\}$ of section~\ref{sec:cure:no-signaling} for each group of parties. This basis contains three types of elements: $\mu$, $\lambda$ and $\nu$ (for groups involving more than one party, any basis element containing a $\nu$ vector for one of its party is of type $\nu$). In this basis, identity~\eqref{eq:defineComposite} thus takes the form (written here in a table form similar to Figure~\ref{fig:coefficientsmatrix})
\begin{equation}\label{eq:table}
\left(\begin{array}{c c c}
\gamma(\mu,\mu) & \gamma(\mu,\lambda) & \gamma(\mu,\nu)\\
\gamma(\lambda,\mu) & \gamma(\lambda,\lambda) & \gamma(\lambda,\nu)\\
\gamma(\nu,\mu) & \gamma(\nu,\lambda) & \gamma(\nu,\nu)
\end{array}\right)
= \kappa \mu_{ax}\mu_{by} + 
\left(\begin{array}{c}
\gamma_A(\mu)\\
\gamma_A(\lambda)\\
\gamma_A(\nu)\\
\end{array}\right)
\left(\begin{array}{ccc}
\gamma_B(\mu) & \gamma_B(\lambda) & \gamma_B(\nu)
\end{array}\right),
\end{equation}
where the different $\gamma$s might denote vectors or matrices when their arguments include either $\lambda$ or $\nu$.
Clearly, setting the last lines and columns to 0 on the LHS and RHS of~\eqref{eq:table} preserves equality. This also defines expressions $\overline{\c}_{abxy}$, $\overline{\c}'_{ax}$ and $\overline{\c}''_{by}$ through Eq.~\eqref{eq:cbar}, thus concluding the proof.
\end{proof}


The previous proposition ensures that when checking for the decomposability of a Bell expression, one does not need to consider the adjunction of terms containing elements of the type $\nu$ to the expression. The addition of a constant, however can be necessary to decompose a Bell expression. For instance, $P_A(1|1)+P_B(1|1)+P_{AB}(11|11)$ is only a tensor product after addition of the constant 1. Thus exactly one parameter needs to be chosen in order to verify whether an expression can be decomposed as a tensor product: a shift by a constant. Let us show that an expression admits a decomposition only for one value of this constant.

\begin{proposition}\label{propositionSingleConstant}
Consider a non-i/o-lifted multipartite Bell expression. If the addition of a constant $\kappa$ allows one to factor a group of parties from this expression, and the addition of a constant $\kappa'$ allows one to factor another group of parties, then one must have $\kappa=\kappa'$.
\end{proposition}
\begin{proof}
The statement is trivial in the bipartite case since factoring one party then automatically factors the other one, and a party can be factored only for one value of the constant. Let us thus consider a Bell expression $\c_{abcxyz}$ with three groups of parties.

Following the parametrization of Section~\ref{sec:cure:no-signaling}, this Bell expression can be written as a tensor of $\gamma$ components which we denote by $\Gamma_{ijk}$: here each index corresponds to one group of parties. Thanks to Proposition~\ref{propositionSig}, we know that we can choose all free components of $\Gamma_{ijk}$ equal to zero, except possibly the normalization $\Gamma_{111}$.

That one group of parties, say A, can be factored in $\Gamma_{ijk}$ after addition of the constant $\kappa$ means that there exist some tensors $\Gamma^A_i$ and $\Gamma^{BC}_{jk}$ such that
\begin{equation}
\Gamma_{ijk} = \kappa \delta_{i,1}\delta_{j,1}\delta_{k,1} + \Gamma^A_i \cdot \Gamma^{BC}_{jk}.
\end{equation}
Simliarly, one has $\Gamma^C_k$ and $\Gamma^{AB}_{ij}$ such that
\begin{equation}
\Gamma_{ijk} = \kappa' \delta_{i,1}\delta_{j,1}\delta_{k,1} + \Gamma^{AB}_{ij} \cdot \Gamma^C_{k}
\end{equation}
if C can be factored out after addition of the constant $\kappa'$.

The fact that $\Gamma$ is not a lifting under input and outputs implies that there exist an index, denote it $(i,j,k)=(2,1,2)$ such that $\Gamma_{212}\neq 0$. Let us write
\begin{align}\label{eq:explicit4}
\Gamma^A_1\cdot \Gamma^{BC}_{11}-\kappa &= \Gamma^{AB}_{11}\cdot \Gamma^C_1 - \kappa' ,&  \Gamma^A_1\cdot \Gamma^{BC}_{12} &= \Gamma^{AB}_{11}\cdot \Gamma^C_2,\nonumber\\
\Gamma^A_2\cdot \Gamma^{BC}_{11} &= \Gamma^{AB}_{21}\cdot \Gamma^C_1 ,&  \Gamma^A_2\cdot \Gamma^{BC}_{12} &= \Gamma^{AB}_{21}\cdot \Gamma^C_2.
\end{align}
Since $\Gamma_{212}\neq 0$, we find that $\Gamma^A_2,\ \Gamma^{AB}_{21},\ \Gamma^{BC}_{12},\ \Gamma^C_2 \neq 0$. Moreover, if $\Gamma^A_1=0$ then $\Gamma^{AB}_{11}=0$ and we directly obtain $\kappa=\kappa'$. Similarly if $\Gamma^C_1=0$. So we are only interested in the case $\Gamma^A_1,\Gamma^C_1\neq 0$, which has $\Gamma^{AB}_{11},\Gamma^{BC}_{11}\neq 0$. Therefore we can write
\begin{equation}
\Gamma^A_2=\frac{\Gamma^{AB}_{21}}{\Gamma^{BC}_{11}}\Gamma^C_1,\ \ \Gamma^C_2=\frac{\Gamma^{BC}_{12}}{\Gamma^{AB}_{11}}\Gamma^A_1
\end{equation}
and thus
\begin{equation}
\begin{split}
\Gamma^C_1 \frac{\Gamma^{AB}_{21}}{\Gamma^{BC}_{11}}\Gamma^{BC}_{12} = \Gamma^A_1 \frac{\Gamma^{BC}_{12}}{\Gamma^{AB}_{11}}\Gamma^{AB}_{21} \quad \Leftrightarrow \quad \Gamma^C_1 \Gamma^{AB}_{11} = \Gamma^A_1 \Gamma^{BC}_{11}.
\end{split}
\end{equation}
Compared to the first equation of~\eqref{eq:explicit4}, this shows that $\kappa=\kappa'$.
\end{proof}

One consequence of this proposition is that when writing an expression in a table form across a bipartition (like in the LHS of Eq.~\eqref{eq:table}), it is only possible for an expression to be a product across this bipartition if the rank of this matrix is smaller or equal to 2.

Proposition~\ref{propositionSingleConstant} also guarantees that the order of groups of parties according to which one tests for a tensor structure does not matter: all give the same result. Thus, after the identification of its tensor structure, a composite Bell expression must be of the following form:
\begin{equation}\label{eq:typicalComposition}
\c \sim \kappa_1 + \c_{A}\otimes \c_{BC} \otimes (\kappa_2 + \c_{DE}\otimes \c_{FGH}),
\end{equation}
where we made the tensor product explicit and indices on the Bell expressions indicate the parties they involve. Note that in this decomposition, expressions are only defined up to a constant, which is arbitrary: changing $\c_{DE}$ for $\c_{DE}/2$ and $\c_{FGH}$ for $2\c_{FGH}$ keeps Eq.~\eqref{eq:typicalComposition} unchanged. The magnitude of these constants can however be fixed when dealing with rational Bell expressions by requiring that the coefficients of all expressions be integers with greatest common divisors equal to 1. The only remaining freedom then lies in the choice of sign: $\c_{DE}\otimes\c_{FGH}$ is equaliy valid as $(-\c_{DE})\otimes(-\c_{FGH})$.

\subsection{Properties of composite inequalities}
Bell expressions can be bounded with respect to convex sets of correlations. Let us show that when these sets satisfy condition~\eqref{eq:conditionality} of the main text, the bound of a composite expression is inherited from the bounds of its sub-expressions.

\begin{proposition}\label{propositionBound}
If the following bounds hold for Bell expressions $\c'_{ax}$ and $\c''_{by}$ with respect to a no-signaling model satisfying Eq.~\eqref{eq:conditionality}:
\begin{eqnarray}
\b'_- \leq \sum_{ax}\c'_{ax}P(a|x) \leq \b'_+ \label{eq:boundscp}\\
\b''_- \leq \sum_{by}\c''_{by}P(b|y) \leq \b''_+,\label{eq:boundscpp}
\end{eqnarray}
then expression $\c_{abxy}=\kappa\mu_{ax}\mu_{by}+\c'_{ax}\c''_{by}$ satisfies
\begin{equation}\label{eq:twoBoundsOnCompositeExpr}
\kappa + \min (\b'_- \b''_-,\ \b'_- \b''_+,\ \b'_+ \b''_-,\ \b'_+ \b''_+) \ \leq\ \sum_{abxy}\c_{abxy}P(ab|xy) \ \leq\ \kappa + \max (\b'_- \b''_-,\ \b'_- \b''_+,\ \b'_+ \b''_-,\ \b'_+ \b''_+).
\end{equation}
Moreover, these bounds are tight if the bounds on $\c'_{ax}$ and $\c''_{by}$ are.

\end{proposition}
\begin{proof}
Let us consider the value of the expression defined by $\c_{abxy}$:
\begin{equation}
\begin{split}
\sum_{abxy}\c_{abxy}P(ab|xy) &= \sum_{abxy}\kappa \mu_{ax}\mu_{by} P(ab|xy) + \sum_{abxy}\c'_{ax}\c''_{by}P(ab|xy)\\
&= \kappa + \sum_{ax} \c'_{ax} P(a|x) \sum_{by} \c''_{by} P(b|y,ax),
\end{split}
\end{equation}
where we used the no-signaling condition. Now because of condition~\eqref{eq:conditionality} and~\eqref{eq:boundscpp}, we must have that $\b''_- \leq \sum_{by} \c''_{by} P(b|y,ax) \leq \b''_+$ for all $a,x$. Moreover, these bounds must be tight if the ones for unconditioned correlations are. This gives the lower and upper bound of \eqref{eq:twoBoundsOnCompositeExpr} by considering any possible combination of bounds on this expression together with bounds on $\c'_{ax}$. These combined bounds are achievable by construction.
\end{proof}

From proposition~\ref{propositionBound} and the fact that condition~\eqref{eq:conditionality} is satisfied for the sets of local, quantum and no-signaling correlations, we deduce that the local, quantum and no-signaling bounds of composite inequalities are inherited from the corresponding bounds on each of their component expression.

On the other hand, given a bound $\b$ on $\c$ for models satisfying~\eqref{eq:conditionality}, Eq.~\eqref{eq:twoBoundsOnCompositeExpr} also puts constraints on the possible bounds that $\c'$ and $\c''$ can have. For instance, if the bounds on one of the sub-expressions are known (because it is a single-party expression for instance) then bounds can be deduced for the second expression. A particular example of this is when $\b'_-$ and $\b'_+$ are known and strictly positive. In this case, it is clear that $\b=\kappa+\max (\b'_-\b''_-,\b'_-,\b''_+,\b'_+,\b''_-,\b'_+,\b''_+)=\kappa+\b'_+\b''_+$, which allows one to deduce the value of $\b''_+$.

Finally, when the sets of correlations of interest admits facets (e.g. it is a polytope), any two facets can be combined to produce a new facet.

\begin{proposition}
If $\sum_{ax}\c'_{ax}P(a|x)\leq\b'$ and $\sum_{by}\c''_{by}P(b|y)\leq\b''$ define facets for a model satisfying Eq.~\eqref{eq:conditionality}, then 
\begin{equation}\label{eq:compositeFacet}
\sum_{abxy}\c_{abxy}P(ab|xy)\leq 0   \quad \text{with}\quad   \c_{abxy} = -(\c'_{ax}-\b\mu_{ax})(\c''_{by}-\b\mu_{by})
\end{equation}
is also a facet for the same model.
\end{proposition}
\begin{proof}
To demonstrate the facet property of the composite inequality, we construct a set of $d=d_1 d_2 + d_1 + d_2$ linearly independent probability vectors which saturate the inequality. Here $d_1$ denotes the dimension of the normalized no-signaling subspace $\omega_1$ for the first group of $n_1$ parties, i.e. $d_1 = (1+m(k-1))^{n_1}-1$ for a homogeneous scenario, and $d_2$ denotes the dimension of $\omega_2$ for the second group of parties. The dimension of the normalized no-signaling space $\omega$ is $d=(1+d_1)(1+d_2)-1=d_1 d_2 + d_1 + d_2$.

The fact that $\sum_{ax}\c'_{ax}P(a|x)\leq\b'$ is a facet in $\omega_1$ implies that there exist $d_1$ linearly independent probability vectors $q_{i,ax}$ which belong to the first set of correlations and saturate this inequality. Since $\omega_1$ satisfies the normalization condition, it lives in an affine subspace of the full proobability space. The set $\{q_{i,ax}\}_i$ can thus be completed by one vector from the set of correlations allowed by the model for the first $n_1$ parties to form a basis $\{r_{i,ax}\}_i\supset\{q_{i,ax}\}_i$ of a space containing $\omega_1$.

Similarly, we can define sets $\{q_{j,by}\}_j$ and $\{r_{j,by}\}_j\supset\{q_{j,by}\}_j$ of respectively $d_2$ and $d_2+1$ linearly independent probability vectors. All these vectors belong to the set of correlations for the second group of $n_2$ parties, and every $q_{j,by}$ saturates the second inequality, i.e. satisfies $\sum_{by}\c''_{by} q_{j,by} = \b''\ \forall\ j$.

Since the considered model satisfies condition~\eqref{eq:conditionality}, conditional probabilities for the second group of parties $P(b|y,ax)$ in $\omega$ can take any value that $P(b|y)$ is allowed to take within the set $\omega_2$. Therefore, the tensor product of any two points $r_{i,ax}r_{j,by}$ defines a valid point for the set of correlations in $\omega$, for all $i$ and $j$.

One can check that the inequality is saturated by any strategy of the form $q_{i,ax}r_{j,by}$ or $r_{i,ax}q_{j,by}$. This thus define $d_1(d_2+1)+(d_1+1)d_2$ points saturating~\eqref{eq:compositeFacet} in $\omega$. Since all points $r_{i,ax}r_{j,by}$ are linearly independent (they form a basis of a $(d_1+1)(d_2+1)$ dimensional space), these vectors are all linearly independent except for duplicates. There are $d_1 d_2$ duplicates, of the form $q_{i,ax}q_{j,by}$. This leaves us with $d_1d_2+d_1+d_2$ linearly independent points in the allowed set of correlations which saturate~\eqref{eq:compositeFacet}. This inequality is thus a facet.



%

\end{proof}

%
%

\end{document}